\documentclass[11pt]{article}
\usepackage{fullpage}
\usepackage[utf8]{inputenc}
\usepackage{braket}
\usepackage{amsthm}
\usepackage{amsmath, amssymb}
\usepackage[pdftex]{graphicx}
\usepackage{xcolor}
\usepackage{algorithm}
\usepackage{mathtools}
\usepackage{algpseudocode}
\usepackage{latexsym}
\usepackage{optidef}

\usepackage[pagebackref]{hyperref}
\hypersetup{
    linktocpage,
    colorlinks=true,
    citecolor=blue,
    linkcolor=magenta,
    urlcolor=orange,
}

\newcommand{\QMA}{\mathsf{QMA}}
\newcommand{\Q}{\mathsf{Q}}

\newcommand{\D}{\mathsf{D}}
\newcommand{\R}{\mathsf{R}}
\newcommand{\EQ}{\mathsf{EQ}}
\newcommand{\poly}{\mathrm{poly}}

\newcommand{\E}{\mathbb{E}}

\newcommand{\email}[1]{\href{mailto:#1}{\texttt{#1}}}

\usepackage[capitalise,nameinlink]{cleveref}
\Crefname{lemma}{Lemma}{Lemmas}
\Crefname{fact}{Fact}{Facts}
\Crefname{theorem}{Theorem}{Theorems}
\Crefname{corollary}{Corollary}{Corollaries}
\Crefname{claim}{Claim}{Claims}
\Crefname{example}{Example}{Examples}
\Crefname{problem}{Problem}{Problems}
\Crefname{definition}{Definition}{Definitions}
\Crefname{notation}{Notation}{Notations}
\Crefname{assumption}{Assumption}{Assumptions}
\Crefname{subsection}{Section}{Sections}
\Crefname{section}{Section}{Sections}
\Crefformat{equation}{(#2#1#3)}

\newtheorem{theorem}{Theorem}
\newtheorem{definition}{Definition}
\newtheorem{lemma}[theorem]{Lemma}

\newtheorem{fact}{Fact}
\newtheorem{corollary}[theorem]{Corollary}
\newtheorem{conjecture}{Conjecture}
\newtheorem{proposition}[theorem]{Proposition}

\theoremstyle{remark}
\newtheorem{remark}{Remark}

\title{Does there exist a quantum fingerprinting protocol \\ without coherent measurements?}

\author{Atsuya Hasegawa\thanks{Graduate School of Mathematics, Nagoya University. \ Email: \email{atsuya.hasegawa@math.nagoya-u.ac.jp}} \and Srijita Kundu\thanks{Institute for Quantum Computing and Department of Combinatorics and Optimization, University of Waterloo. \ Email: \email{srijita.kundu@uwaterloo.ca}} \and Fran{\c{c}}ois Le Gall\thanks{Graduate School of Mathematics, Nagoya University. \ Email: \email{legall@math.nagoya-u.ac.jp}} \and Harumichi Nishimura\thanks{Graduate School of Informatics, Nagoya University. \ Email: \email{hnishimura@i.nagoya-u.ac.jp}} \and Qisheng Wang\thanks{School of Informatics, University of Edinburgh. \ Email: \email{QishengWang1994@gmail.com}}}

\date{}

\begin{document}

\maketitle

\begin{abstract}
Buhrman, Cleve, Watrous, and de Wolf (PRL 2001) discovered the quantum fingerprinting protocol, which is the quantum SMP protocol with $O(\log n)$ qubits communication for the equality problem. In the protocol, Alice and Bob create some quantum fingerprints of their inputs, and the referee conducts the SWAP tests for the quantum fingerprints. Since $\Omega(\sqrt{n})$ bits communication is required with the classical SMP scheme for the equality problem first shown by Newman and Szegedy (STOC 1996), there exists an exponential quantum advantage in the amount of communication.

In this paper, we consider a setting in which the referee can do only incoherent measurements rather than coherent measurements including the SWAP tests. We first show that, in the case of one-way LOCC measurements, $\Omega(\sqrt{n})$ qubits communication is required. To prove the result, we derive a new method to replace quantum messages by classical messages and consider a reduction to the optimal lower bound in the hybrid SMP model where one message is quantum and the other is classical, which was first shown by Klauck and Podder (MFCS 2014). Our method uses the result of Oszmaniec, Guerini, Wittek, and Ac{\'i}n (PRL 2017), who showed that general POVM measurements can be simulated by randomized projective measurements with small ancilla qubits, and Newman's theorem.

We further investigate the setting of quantum SMP protocols with two-way LOCC measurements, and derive a lower bound against some restricted two-way LOCC measurements. To prove it, we revisit the technique to replace quantum messages by classical deterministic messages introduced by Aaronson (ToC 2005) and generalized by Gavinsky, Regev, and de Wolf (CJTCS 2008), and show that, using the deterministic message, the referee can simulate the two-way LOCC measurements.
\end{abstract}

\clearpage

\tableofcontents

\section{Introduction}

\subsection{Background}

In 1979, the seminal paper by Yao \cite{Yao79} introduced the concept of communication complexity. The big picture behind this concept is, when inputs are distributed to some local processors and they coordinate to compute a function depending on distributed inputs, to identify the minimum amount of communication between them (rather than considering the computational resource in each local party). Yao introduced several communication models, including two-party communication protocols, and one important communication model is the Simultaneous Message Passing (SMP) model (referred to as $``1 \rightarrow 3 \leftarrow 2"$ in the paper). In the SMP model, Alice has a part of an input $x$, Bob has the other part of the input $y$, and a third party, whom we usually call the referee, calculates some Boolean function $F(x,y)$ by receiving two messages from Alice and Bob simultaneously. The complexity in the SMP protocols is the sum of the length of the two messages. One fundamental function is the equality function ($\EQ_n$ where $n$ stands for the size of inputs) to check if the inputs are the same or not (named ``the identification function" in the paper), and Yao asked what the complexity in the SMP model is for $\EQ_n$. There is a trivial upper bound of $O(n)$ for the complexity while it can be reduced much when Alice and Bob are allowed to share randomness by sending hashing values of their inputs. 
However, without shared randomness, the SMP model is a weak communication model, and the complexity was unknown.

In the 1990s, the question by Yao was solved. First, Ambainis \cite{Amb96} constructed an SMP protocol of cost $O(\sqrt{n})$ that exploits properties of good classical error correction codes. Subsequently, Newman and Szegedy \cite{NS96} proved a matching lower bound of $\Omega(\sqrt{n})$. Their results were simplified and generalized by Babai and Kimmel \cite{BK97}, who showed that the randomized and
deterministic complexities can be at most quadratically far apart for any function in this model. Let us denote by $\R^{||}(\EQ_n)$ the complexity in the classical (randomized) SMP model for $\EQ_n$, and the results are summarized as follows.

\begin{theorem}[\cite{Amb96,NS96,BK97}]
    $\R^{||}(\EQ_n) = \Theta (\sqrt{n})$.
\end{theorem}

Yao \cite{Yao93} also introduced a quantum analog of the two-party communication complexity, and subsequent pioneering works \cite{Kre95,CB97,BvDHT99,CvDNT98,BCW98,Raz99,BdW01} investigated the power and limit of quantum communication complexity. In 2001, Buhrman, Cleve, Watrous, and de Wolf \cite{BCWdW01} considered a quantum analog of the SMP model where Alice and Bob can send some quantum messages and the referee conducts some quantum algorithm over two messages to calculate a function. They discovered the quantum fingerprinting protocol, which is the quantum SMP protocol with $O(\log n)$ qubits of communication for the equality problem. Since $\Omega(\sqrt{n})$ bits of communication is required for the classical SMP model, there exists an exponential quantum advantage. In the \cite{BCWdW01} protocol, Alice and Bob create some quantum fingerprints of their inputs as a superposition of a codeword of their inputs with some good classical error correction codes, and the referee conducts a SWAP test \cite{BBD+97,BCWdW01} on the quantum fingerprints. Note that this bound is tight because the quantum one-way communication complexity for $\EQ_n$ is $\Theta(\log n)$ and the one-way communication is a stronger model than the SMP model \cite{BdW01,BCWdW01}. Let us denote by $\Q^{||}(\EQ_n)$ the complexity in the quantum SMP model for $\EQ_n$, and the result is denoted as follows.

\begin{theorem}[\cite{BCWdW01}]
    $\Q^{||}(\EQ_n) = \Theta(\log n)$.
\end{theorem}

Gavinsky, Regev, and de Wolf \cite{GRdW08} first considered a classical and quantum hybrid SMP protocol where one party can send a quantum message, while the other party can send only a classical message. They showed that $\Omega(\sqrt{n/\log n})$ qubits or bits communication is required for $\EQ_n$ in this model. Klauck and Podder \cite{KP14}, and Bottesch, Gavinsky, and Klauck \cite{BGK15} showed a tight lower bound to match the upper bound by Ambainis \cite{Amb96}. Let us denote by $\R \Q^{||}(\EQ_n)$ the complexity in the hybrid SMP model for the equality problem. Then, the tight lower bounds are described as follows.

\begin{theorem}[\cite{KP14,BGK15}]\label{thm:lower_bound_hybrid}
    $\R \Q^{||}(\EQ_n) = \Theta(\sqrt{n})$.
\end{theorem}

\subsection{Our setting and result}

With our current technology in quantum computing, it is still hard to implement some {\em coherent measurements} over two quantum states 
(i.e., measurements done jointly over them) with high accuracy, and the SWAP test is one of the canonical coherent measurements. In this paper, we pose the following question: 
\newline
\newline
\centerline{\emph{Does there exist an efficient quantum SMP protocol for $\EQ_n$ without coherent measurements?}}
\newline

Here \emph{efficient} means $\poly (\log n)$-qubit cost which exhibits an exponential advantage over the classical SMP models. Let us denote by $\Q^{||,\mathbb{M}}(R)$ the complexity in the quantum SMP for a relation $R$ where the referee is restricted to performing measurements from a measurement class $\mathbb{M}$. See \cref{sec:measurement} for formal definitions of some measurement classes. Note that a function is a special case of a relation in which there is a unique solution for each input, and we also denote by $\Q^{||,\mathbb{M}}(F)$ the complexity for a function $F$. It is known that the SWAP test is a projection to the symmetric subspace and therefore the corresponding measurement operator is in a measurement class SEP\footnote{Other works may denote by $\mathrm{SEP}$ the set of $M$ such that $M$ and $I-M$ can be decomposed as separable operators, and by $\mathrm{SEP}_\mathrm{YES}$ the set of $M$ such that only $M$ has to satisfy the operator is separable. See \cref{sec:measurement} for formal definitions.} \cite{HM13}. Therefore, we can rewrite the result of \cite{BCWdW01} as follows.

\begin{theorem}[\cite{BCWdW01}]
    $\Q^{||,\mathrm{SEP}}(\EQ_n) = \Theta(\log n)$. 
\end{theorem}

What is a measurement that is not coherent (or {\em incoherent}) one? Operationally, it is considered as the so-called {\em LOCC (Local Operations and Classical Communication)} measurements, which is well studied in the quantum information community (see, e.g., \cite{CLM+14}). This type of measurement over two states is a sequence of measurements over either of the states 
which can use the previous measurement results via classical communication between the states. 
Thus, our question can be described as follows:
\newline
\newline
\centerline{\emph{Does there exist an efficient quantum SMP protocol for $\EQ_n$ with LOCC measurements?}}
\newline

We first investigate the power of a quantum SMP model with a subclass of LOCC measurements, the quantum {\em one-way-LOCC} SMP model. In this model, the referee performs a measurement on one quantum message and then performs a measurement on the other quantum message based on the first measurement result. The final measurement outcome is expected to be a solution for each input. This is easier to implement than coherent measurements as well as general (two-way) LOCC measurements since we do not need any quantum memory for keeping the residue states after the measurements. Let us denote by $\Q^{||,\mathrm{LOCC}_1}(R)$ the message complexity in the model for $R$. Although it is a stronger model than the classical and quantum hybrid SMP model, we provide a new method to replace a quantum message by a classical message and prove the quantum advantage is small. Note that $\Q^{||,\mathrm{LOCC}_1}(R) \leq \R\Q^{||}(R)$ because a quantum one-way-LOCC SMP model is a stronger model than a hybrid SMP model.

\begin{theorem}[Informal version of \cref{thm:one-way_LOCC}]\label{thm:intro_one-way_LOCC}
    For any relation $R \subseteq \{0,1\}^n \times \{0,1\}^n \times \{0,1\}^*$,
    \[
        \R\Q^{||}(R) \leq 2 \Q^{||,\mathrm{LOCC}_1}(R) + O(\log n).
    \]
\end{theorem}

Combined with \cref{thm:lower_bound_hybrid}, the above theorem shows that there exists no efficient protocol in the quantum one-way-LOCC SMP model.

\begin{corollary}[\cref{cor:EQ}]
    $\Q^{||,\mathrm{LOCC}_1}(\EQ_n) = \Omega(\sqrt{n})$.
\end{corollary}

Moreover, we can describe a more general result for Boolean functions. For any Boolean function $F:\{0,1\}^n \times \{0,1\}^n \rightarrow \{0,1\}$, it is known that $\R^{||}(F) \leq \D^{||}(F)^2$ \cite{BK97} ($\D^{||}$ denotes the complexity in the deterministic SMP model), and $\R\Q^{||}(F) \leq \R^{||}(F)^2$ \cite{GRdW08}. Combining with \cref{thm:intro_one-way_LOCC}, we can claim that, for any Boolean function which is hard for classical deterministic SMP models, the advantage of quantum one-way-LOCC SMP models over classical deterministic SMP models is at most quartic.

\begin{corollary}[\cref{cor:quartic}]
    For any Boolean function $F:\{0,1\}^n \times \{0,1\}^n \rightarrow \{0,1\}$, if $\D^{||}(F) = \Omega(n^c)$ for a constant $c > 0$, then $\Q^{||,\mathrm{LOCC_1}} (F) = \Omega(n^{\frac{c}{4}})$.
\end{corollary}

The one-way LOCC measurement is a good intermediate setting to analyze because the direction of the information flow in the measurement is indeed one-way. In the more general setting of two-way LOCC measurements, the measurement results go back and forth between two parties, which gives a lot of freedom to do local measurements, and makes the analysis difficult to obtain lower bounds. For the complexity class $\QMA(2)$ \cite{KMY03,HM13}, which is an analogue of $\QMA$ but the quantum polynomial-time verifier receives two separate quantum states for certification, some results have been obtained in a similar direction to investigate what problems the verifier can verify when the measurement ability is restricted. Although the class $\QMA^{\mathbb{M}}(2)$ under the measurement restriction collapses to $\QMA$ when $\mathbb{M}$ is the class of one-way LOCC measurements \cite{BCY11,LS15} (namely, the restriction of one-way LOCC measurements is severe for $\QMA(2)$), to our knowledge, it remains unknown whether it holds or not when $\mathbb{M}$ is the class of LOCC measurements (see \cref{sec:related} for more detailed explanations) for more than 10 years. This situation demonstrates the difficulty of obtaining lower bounds involving two-way LOCC measurements.

In this very challenging setting, we succeed in proving a lower bound in quantum SMP protocols with some restricted two-way LOCC measurements. Here is our result for $\EQ_n$.

\begin{theorem}[Informal version of \cref{thm:two-way_LOCC}]
    Suppose that there exists a quantum two-way-LOCC SMP protocol for $\EQ_n$ whose measurement outcomes are $2$ values and the number of rounds is smaller than $k\log_2 n$ where $k$ is a constant such that $0<k<\frac{2}{3}$. Then, the message size is $n^{\Omega(1)}$.
\end{theorem}

We also prove a lower bound for general problems.

\begin{theorem}[Informal version of \cref{thm:two-way-LOCC_general_relation}]
    For $R \subseteq \{0,1\}^n \times \{0,1\}^n \times \{0,1\}^*$, suppose $\R^{||}(R) = \Omega(n^c)$. Also, suppose that there exists a quantum two-way-LOCC SMP protocol for $R$ whose measurement outcomes are $2$ values and the number of rounds is smaller than $k\log_2 n$ where $k$ is a constant such that $0<k<\frac{c}{3}$. Then, the message size of the quantum two-way-LOCC SMP protocol is $n^{\Omega(1)}$.
\end{theorem}

One may think that this quantum LOCC SMP model is just a weak model, but we show a quantum 2-round LOCC SMP protocol can be exponentially powerful than a quantum one-way-LOCC SMP protocol for a relation problem, albeit with measurements that are not 2-value\footnote{In fact, this protocol can be conducted by a sequence of 2-value measurements. See \cref{remark} for further information.}. The relation problem we considered is a variant of the Hidden Matching Problem introduced by Bar-Yossef, Jayram, and Kerenidis \cite{BYJK08}. 
Let us denote by $\Q^{||,\mathrm{LOCC}}(R)$ the complexity in quantum two-way-LOCC SMP protocols for a relation $R$.

\begin{proposition}[\cref{prop:separation}]
    There exists a relation problem $R$ whose input size is $n$ such that $\Q^{||,\mathrm{LOCC}}(R) = O(\log n)$ and $\Q^{||,\mathrm{LOCC}_1}(R) = \Theta(\sqrt{n})$.
\end{proposition}

We also consider a variant of quantum one-way communication protocols where measurements are not ideal. In the usual setting, a receiver's quantum computer is ``programmable" based on receiver's $2^n$ inputs (see \cref{def:Q_one-way} for a formal definition). In an actual setting, it is hard to implement a measurement from a large number of varieties with high accuracy. Instead of this situation, let us consider the case where we implement some complex and accurate measurement whose outcome can be a large classical string, but the measurement bases are fixed for any input, and we cannot keep residue states after the first measurement. In such a situation, are quantum messages more useful than classical messages? We show a negative answer for the question above by the same technique derived to have a lower bound in the quantum one-way-LOCC SMP models for $\EQ_n$. Let us call the protocol a quantum incoherent one-way communication complexity and denote by $\Q^{1,\perp}(R)$ the quantum incoherent one-way communication complexity for $R$ (see \cref{def:Q_one-way_incoherent} for a formal definition). Let us also denote by $\R^1(R)$ the randomized one-way communication complexity for $R$. Note that $\Q^{1,\perp}(R) \leq \R^1(R)$ from their definitions.

\begin{corollary}[\cref{cor:incoherent}]\label{cor:incoherent_intro}
    For any relation $R \subseteq \{0,1\}^n \times \{0,1\}^n \times \{0,1\}^*$, 
    \[
    \R^1(R) \leq 2 \Q^{1,\perp}(R) + O(\log n).
    \]
\end{corollary}

There are some known results to show that we can replace quantum messages by classical messages in one-way communication complexity. Here we consider a (possibly partial) function $F : X \times Y \rightarrow \{0,1\}$. Let us denote by $\Q^{1,*}(F)$ a variant of quantum one-way communication complexity where Alice and Bob share (unlimited) prior entanglement. We also consider classical one-way communication complexity and denote by $\D^1(F)$ (${\rm resp.}$ $\R^1(F)$) the classical deterministic (${\rm resp.}$ randomized) one-way communication complexity for $F$.

\begin{theorem}[Theorem 3.4 in \cite{Aar05}]
    $\D^1(F) = O(\log (|Y|) \Q^1(F) \log \Q^1(F) )$.
\end{theorem}

\begin{theorem}[Theorem 1.4 in \cite{Aar07}]
    $\R^1(F) = O( \log (|Y|) \Q^1(F))$.
\end{theorem}

\begin{theorem}[Theorem 3.1 in \cite{KP14}]
    $\D^1(F) = O( \log (|Y|) \Q^{1,*}(F))$.
\end{theorem}

By comparing \cref{cor:incoherent_intro} with these results, we can claim that the quantum incoherent one-way communication protocols are much weaker protocols than usual quantum one-way communication protocols.

\subsection{Our techniques}

\subsubsection*{Against one-way LOCC measurements}

For clarity, in quantum LOCC SMP protocols, let us decompose the referee into two referees $\rm{Ref}_A$ and $\rm{Ref}_B$, who receive quantum messages from Alice and Bob, respectively.
In the quantum one-way-LOCC SMP model, $\rm{Ref}_A$ first measures Alice's message and, based on the measurement result, $\rm{Ref}_B$ measures Bob's message. Our proof strategy is to make Alice do the measurement of $\rm{Ref}_A$ on Alice's message state and send it to the referee. Then, the SMP protocol is the classical and quantum hybrid SMP model.

However, there is an issue that comes up in implementing this strategy. The measurement by $\rm{Ref}_A$ can be done by adding many ancilla qubits to Alice's message state, and the measurement result can be a large classical string. Therefore, we cannot consider a direct reduction to the lower bound in the hybrid SMP models. To overcome this, we use a result of Oszmaniec, Guerini, Wittek, and Ac{\'i}n \cite{OGWA17}, who showed that general POVM measurements can be simulated by randomized projective measurements with few ancilla qubits (see \cref{lem:OGWA17} for a formal description). Since the number of outcomes is small if the measurements are projective measurements with a small number of ancilla qubits, the result may lead to a nice reduction.

However, there is still another issue that comes up. To recover the original measurement result and exactly simulate $\rm{Ref}_A$, 
the referee needs to know which measurement Alice made. If the number of measurements is limited, then we could make Alice send which measurement she did. The number of measurements corresponds to the number of unitary channels that must be in a convex combination to obtain any mixed unitary channel \cite{CD04,DPP05}. Moreover, the number is related to Carath{\'e}odory’s theorem \cite{Car11} and generally large (see Proposition 4.9 in \cite{Wat18}). Therefore, for a valid reduction, Alice and the referee need to share a probabilistic mixture over a large number of alphabets to know which measurement Alice did.

Fortunately, we can show that a variant of Newman's theorem \cite{New91} holds in this setting (\cref{lem:newman}), and we can remove the generally large shared randomness to a small ($O(\log n)$ bits) private randomness and make Alice send it to the referee. Therefore, we have a valid reduction from a quantum one-way-LOCC SMP protocol to the hybrid SMP protocol for any relation. Since $\EQ_n$ is hard for the hybrid SMP protocols, it is also hard for quantum one-way-LOCC SMP protocols.

\subsubsection*{Against two-way LOCC measurements}

Aaronson \cite{Aar05} introduced a method to replace quantum messages by classical deterministic messages in one-way communication protocols. His construction exploits the property that if a 2-value measurement accepts a state with high probability, after the measurement, one can recover a state that is close to the original state just before the measurement (see Lemma 2.2 in \cite{Aar05}). Theorem 3.4 in \cite{Aar05} assumes that the measurement probability is close to 0 or 1.  Gavinsky, Regev and de Wolf \cite{GRdW08} proved a similar argument without this assumption. See also Theorem 4 in \cite{Aar18}.

\begin{theorem}[Theorem 5 in \cite{GRdW08}]
    Suppose Alice has the classical description of an arbitrary $q$-qubit density matrix $\rho$, and Bob has $2^c$ $2$-value POVM operators $\{E_b\}_{b \in \{0,1\}^c}$. For any constant $\delta >0$, there exists a deterministic message of $O(qc\log q)$\footnote{This asymptotic order hides some large constant depending on $\delta$ and holds for any small constant $\delta$. See \cref{lem:replace} for a formal statement.} bits from Alice that enables Bob to output values $p'_b$ satisfying that $|p_b-p'_b| \leq \delta$ simultaneously for all $b \in \{0,1\}^c$ where $p_b = \mathrm{tr}(E_b \rho)$.
\end{theorem}

Our first observation is that the dependence of the message size on the number of measurements in this technique is very good: using a deterministic message, the referee can guess the acceptance probabilities for \emph{all} $2^n$ inputs (this property is also noted in Section 5 in \cite{GRdW08}). The receiver in one-way communication only needs to use one value to estimate acceptance probability. Our idea is to use all the values to simulate quantum two-way-LOCC SMP protocols with not too many rounds.

However, we cannot directly apply the technique because during an LOCC measurement, a number of 2-value measurements will be performed sequentially on the post-measurement state from the previous measurement. To make the idea work, we show that the measurement outcomes at each LOCC iteration can be estimated by taking some ratios of values the referee can estimate from deterministic messages (\cref{lem:ratio}). Then, by carefully evaluating the error caused by the simulation, we have the desired result.

\subsubsection*{Separation between quantum one-way-LOCC and two-way-LOCC SMP protocols}

Bar-Yossef, Jayram and Kerenidis \cite{BYJK08} introduced the Hidden Matching problem to separate classical and quantum one-way communication complexity (this was the first exponential separation between them). By limiting the number of matchings by $\Theta(n)$, they also introduced Restricted Hidden Matching problem and showed an exponential separation between classical and hybrid SMP models.

We introduce a further variant of the Restricted Hidden Matching problem we call Double Restricted Hidden Matching problem, and make both parties send quantum messages, i.e., the problem is hard for hybrid SMP schemes. It is shown that the problem is easy if both quantum messages can be measured adaptively based on the other message, i.e., for quantum two-way-LOCC SMP protocols. Moreover, since we know quantum one-way-LOCC SMP protocols and hybrid SMP protocols are almost equivalent (\cref{thm:one-way_LOCC}), we show that the problem is hard for quantum one-way-LOCC SMP protocols.

\subsection{Related works}\label{sec:related}

Variants with LOCC measurements have also been studied for the complexity class $\QMA(2)$, which was first introduced by Kobayashi, Matsumoto, and Yamakami \cite{KMY03}. Harrow and Montanaro \cite{HM13} showed that, by making the verifier conduct the product test, which is a variant of the SWAP test, $\QMA(2)$ contains the extended class $\QMA(k)$ for any $k\geq 2$ where the verifier receives $k$ separate quantum states for certification. Let us denote $\QMA^\mathbb{M}(k)$ as the class $\QMA(k)$ with the verifier restricted to performing measurements from $\mathbb{M}$. As already noted, the SWAP test is a projection to the symmetric subspace and therefore the corresponding measurement operator is in SEP.

\begin{theorem}[Informal version from \cite{HM13}]
For $k \in \mathbb{N}$,
\[
    \QMA^\mathrm{SEP}(2) \supseteq \QMA(k).
\]
\end{theorem}

In contrast, Brandao, Christandl and Yard \cite{BCY11}, and Li and Smith \cite{LS15} showed that if the verifier's ability is restricted to one-way LOCC measurements, the class is contained in $\QMA$. Note that, to our knowledge, it is still unknown whether $\QMA^{\mathrm{LOCC}}(k) \subseteq \QMA$ is or not.

\begin{theorem}[Informal version from \cite{BCY11,LS15}]
For $k \in \mathbb{N}$,
\[
    \QMA^{\mathrm{LOCC}_1}(k) \subseteq \QMA.
\]    
\end{theorem}

Aaronson, Beigi, Drucker, Fefferman and Shor \cite{ABD+09} showed that 3-SAT can be solved by a $\QMA(k)$ protocol with $k=\Tilde{O}(\sqrt{n})$ and messages of length $O(\log n)$ exploiting the property of the SWAP test. Chen and Drucker \cite{CD10} showed that the result can be obtained even when the verifier's ability is restricted to BELL measurements, a subclass of one-way LOCC measurements. 

Recently, the power of incoherent measurements has been explored from sample complexity in the literature of quantum state learning (see, e.g., \cite{BCL20,HKP21,ACQ22,CLO22,CCHL22,CHLL22,CHL+23,CLL24,CGY24}). The tasks they considered include quantum state tomography \cite{HHJ+16,OW16}, state certification \cite{BOW19}, shadow tomography \cite{Aar18,HKP20}, and purity testing \cite{MdW16}.

Anshu, Landau and Liu \cite{ALL22} considered a task where two players
Alice and Bob who have multiple copies of unknown $d$-dimensional states $\rho$ and $\sigma$ respectively, and they want to estimate the inner product of $\rho$ and $\sigma$ up to additive error $\epsilon$ by classical communication. They showed that for this task, which they call distributed quantum inner product estimation, $\Theta( \max \{ \frac{1}{\epsilon^2}, \frac{\sqrt{d}}{\epsilon} \})$ copies are sufficient and necessary in multiple measurement and communication settings. One may consider that their technique for the lower bound on sample complexity can be applied in our setting, but we cannot. To prove it, they derive a connection with optimal cloning by Chiribella \cite{Chi10} and exploit the property that the states $\rho$ and $\sigma$ can be any states over the Haar measure. Although the SWAP test in the quantum fingerprinting protocol is essentially inner product estimation, the encoding of quantum messages on their inputs is discrete. Hence there could be a cleverer encoding compared to Haar-random states, and we are unable to use techniques from \cite{ALL22}.

\subsection{Outlook}

In this paper, we introduce and investigate the setting in which the referee can conduct only incoherent measurements in the SMP model. We show that quantum one-way-LOCC SMP models are almost equivalent to classical and quantum hybrid SMP models and thus $\EQ_n$ is hard for quantum one-way-LOCC SMP models. Moreover, if we restrict the number of outcomes and rounds, it is also hard for quantum two-way-LOCC SMP models. We also present a separation between quantum one-way and two-way-LOCC SMP models by a relation problem.

However, despite our efforts, it is still unknown whether there exists a quantum fingerprinting protocol without coherent measurements. We believe that a lower bound would hold with more general two-way LOCC measurements. In quantum information theory, separable measurements are often used to obtain lower bounds for LOCC measurements (this is also the case in \cite{ALL22}). This is because separable measurements include LOCC measurements and they are mathematically simple (see \cref{sec:measurement}). We conjecture that equality remains hard even in the SMP setting where the referee can make measurements where both the accept and reject operators are separable. We use SEP-BOTH to denote this class of measurements.\footnote{This is different from the class of measurements SEP described earlier, in which only the accept operator is required to be separable; the SWAP test is in SEP, but not SEP-BOTH.}

\begin{conjecture}\label{conj}
    $\Q^{||,\mathrm{LOCC}}(\EQ_n) \geq \Q^{||,\mathrm{SEP\text{-}BOTH}}(\EQ_n) \geq \Omega(\sqrt{n})$.
\end{conjecture}

\subsection{Organization}

In \cref{sec:prel}, we introduce the preliminary concepts used in our results. In \cref{sec:one-way_locc}, we prove the lower bound against one-way LOCC measurements. In \cref{sec:two-way_locc}, we prove the lower bound against two-way LOCC measurements. In \cref{sec:separation}, we present a separation between quantum one-way-LOCC and two-way-LOCC SMP protocols by a relation problem. 
\section{Preliminaries}\label{sec:prel}

We assume that readers are familiar with basic notions of quantum computation and information. We refer to \cite{NC10,Wat18,dW19} for standard references.

For a Hilbert (finite-dimensional complex Euclidean) space $\mathcal{H}$ whose dimension is $d$, $\mathcal{B}(\mathbb{C}^d)$ and $\mathcal{D}(\mathbb{C}^d)$ denote the sets of pure and mixed states over $\mathcal{H}$ respectively.

Here are basic probabilistic inequalities.

\begin{fact}[Markov's inequality]\label{fact:markov}
    For a nonnegative random variable $X$ and $a>0$,
    \[
        \Pr[X \geq a] \leq \frac{\E[X]}{a}.
    \]
\end{fact}
\begin{fact}[Hoeffding's bound \cite{Hoe63}]\label{fact:hoeffding}
    Let $x_1,\ldots,x_t$ be independent random variables such that $0 \leq x_i \leq 1$ for $i \in [t]$. Then,
    \[
        \mathrm{Pr}\left[ \left| \frac{1}{t} \left(\sum_{i=1}^t x_i - \mathbb{E} \left[ \sum_{i=1}^t x_i \right] \right) \right| \geq \delta \right] \leq 2e^{-2 \delta^2 t}.
    \]
\end{fact}

\subsection{Communication complexity}
Let us recall some basic definitions of quantum communication complexity. As standard references, we refer to \cite{KN96,RY20} for classical communication complexity and \cite{BCMdW10} for quantum communication complexity and the simultaneous message passing (SMP) model.

The goal in communication complexity is for Alice and Bob to compute a function $F : \mathcal{X} \times \mathcal{Y} \to \{0,1,\perp\}$. We interpret $1$ as ``accept'' and $0$ as ``reject'' and mostly consider $\mathcal{X}=\mathcal{Y}=\{0,1\}^n$. In the computational model, Alice receives an input $x \in \mathcal{X}$ (unknown to Bob), and Bob receives an input $y \in \mathcal{Y}$ (unknown to Alice) promised that $(x,y) \in \mathsf{dom}(F)=F^{-1}(\{0,1\})$. We say $F$ is a total function if $F(x,y) \in \{0,1\}$ for all $x \in \mathcal{X}$ and $y \in \mathcal{Y}$, and $F$ is a partial function otherwise. The bounded error communication complexity of $F$ is defined as the minimum cost of classical or quantum communication protocols to compute $F(x,y)$ with high probability, say $\frac{2}{3}$.

In some settings, it is useful to consider multiple correct outputs for each input. We call such problems relational problems and denote them as $R \subseteq X \times Y \times Z$ where $(x,y) \in X \times Y$ is an input and $z \in Z$ such that $(x,y,z) \in R$ are correct outputs. We can consider a relational problem as a generalization of functions. This is because from a function $F : X \times Y \rightarrow \{0,1, \perp \}$, we can construct a relation $R = \{ (x,y,F(x,y)) \}$ for all $x \in X$ and $y \in Y$ such that $(x,y) \in F^{-1}(\{0,1\})$. Therefore, we state some of our statements for relations, and they also hold for functions.

The simultaneous message passing (SMP) model is a specific model of communication protocols. In this model, Alice and Bob send a single (possibly deterministic, randomized or quantum) message to a referee. The goal of the referee is to output $z$ such that $(x,y,z) \in R$ with high probability, say at least $\frac{2}{3}$. The complexity measure of the protocol is the total amount (bits or qubits) of messages the referee receives from Alice and Bob.  

In this paper, $\D^{||}(R)$ denotes the classical deterministic SMP communication complexity for $R$, $\R^{||}(R)$ denotes the classical randomized SMP communication complexity for $R$, and $\Q^{||}(R)$ denotes the quantum SMP communication complexity for $R$. We also consider a hybrid SMP model in which Alice or Bob sends a quantum message and the other party sends a classical message, and let us denote by $\R\Q^{||}(R)$ the hybrid SMP communication complexity for $R$. There is no shared randomness or entanglement between the parties in these definitions, although we will later introduce definitions with shared randomness. We also denote by $\Q^{||,\mathbb{M}}(R)$ the complexity of the quantum SMP for $R$ where the referee is restricted to performing measurements from $\mathbb{M}$. Several measurement operators will be formally defined in \cref{sec:measurement}.

A basic function considered in communication complexity is the equality function $\EQ_n:~\{0,1\}^n\times\{0,1\}^n\rightarrow \{0,1\}$, 
which is defined as $\EQ_n(x,y)=1$ if $x=y$ and $0$ otherwise.

Let us review the construction of the quantum fingerprint in \cite{BCWdW01}. For a constant $c$, let $E : \{0,1\}^n \to \{0,1\}^N$, where $N=O(n)$, be an error correcting code such that $E(x)$ and $E(y)$ have Hamming distance greater than $cN$ for any distinct $x$ and $y$. Let $E_i(x)$ be the $i$-th bit of $E(x)$. For $x \in \{0,1\}^n$, the quantum fingerprint is defined as $\ket{h_x} = \frac{1}{\sqrt{N}} \sum_i \ket{i}\ket{E_i(x)}$ \cite{BCWdW01}. We can alternatively regard $\ket{h_x}$ as a superposition of hash functions. Let us consider the protocol where Alice and Bob create $\ket{h_x}$ and $\ket{h_y}$ on inputs $x$ and $y$ respectively, and send them to the referee. Since the acceptance probability of the SWAP test is $\frac{1}{2} + \frac{|\braket{h_x | h_y}|^2}{2}$, the referee can distinguish $x=y$ or $x \neq y$ using constantly many copies of quantum fingerprints $\ket{h_x}$ and $\ket{h_y}$ with high probability \cite{BCWdW01}.

Let us recall the formal definition of the standard quantum one-way communication complexity.

\begin{definition}[Quantum one-way communication complexity, $\Q^1(R)$]\label{def:Q_one-way}
    Given a relation $R \subseteq X \times Y \times Z$, Alice has a part of the input $x \in X$ and Bob has the other part of the input $y \in Y$. Alice produces a quantum state $\rho_x$ and sends it to Bob. Then, Bob performs a POVM $\{M^y_{z}\}_{z \in Z}$ on $\rho_x$. We say that a quantum one-way communication protocol computes a relation $R \subseteq X \times Y \times Z$, if for all inputs $(x,y) \in X \times Y$ the measurement outputs $z$ such that $(x,y,z) \in R$ with high probability (say $\frac{2}{3}$). The cost of the protocol is the size of the message $\rho_x$ and we denote by $\Q^1(R)$ the minimum cost for all such protocols to compute $R$.
\end{definition}

We denote by $\Q^{1,*}(R)$ a variant of quantum one-way communication complexity where Alice and Bob share (unlimited) prior entanglement. We also consider classical one-way communication complexity and denote by $\D^1(R)$ (${\rm resp.}$ $\R^1(R)$) the classical deterministic (${\rm resp.}$ randomized) one-way communication complexity for $R$.

Newman \cite{New91} showed that any public-coin randomized communication protocol can be converted into a small private-coin randomized protocol allowing for a small constant error. For a function $F:X \times Y \rightarrow \{0,1\}$ and $c>0$, let us denote by $\R_\gamma(F)$ the communication complexity with private randomness and error $\gamma$, and let us denote by $\R_\gamma^{pub}(F)$ the communication complexity with public unlimited randomness and error $\gamma$. 

\begin{fact}[\cite{New91}]
    Let $F: X \times Y \rightarrow \{0,1\}$ be a function. For every $\epsilon>0$ and $\delta>0$, $\R_{\epsilon+\delta} (F) \leq \R_\epsilon^{pub} (F) + O(\log (\log |X| + \log |Y|) + \log (\frac{1}{\delta}))$.
\end{fact}

\subsection{Measurement in quantum mechanics}
In this subsection, let us recall the definitions of measurements on quantum states.

Quantum measurements are described by a collection of $\{M_m\}_m$ of measurement operators. These are operators acting on the state space of the system being measured. The index $m$ refers to the measurement outcomes that may occur in the experiment. If the state of the quantum system is $\rho$ immediately before the measurement, then the probability that the result $m$ occurs is given by $\mathrm{tr}(M_m^\dagger M_m \rho)$. The state after the measurement is $\frac{M_m \rho M_m^\dagger}{\mathrm{tr}(M_m^\dagger M_m \rho)}$. The measurement operators satisfy the completeness equation, $\sum_m M_m^\dagger M_m = I$.

When we do not care about residual states after measurements, we have a simpler description. Let us define $E_m:=M_m^\dagger M_m$. The collection $\{E_m\}$ is sufficient to determine the probabilities $\mathrm{tr}(E_m\rho)$ of measurement outcomes. This is called a POVM (Positive Operator-Valued Measurement): A POVM on $\mathcal{D}(\mathbb{C}^d)$ with $n$ outcomes is a collection $\mathbf{E}=\{E_1,\ldots,E_n\}$ of non-negative operators satisfying $\sum_i E_i = I$. We denote the set of POVMs on $\mathcal{D}(\mathbb{C}^d)$ with $n$ outcomes by $\mathcal{P}(d,n)$. Given two POVMs $\mathbf{E}=\{E_i\},\mathbf{F}=\{F_i\} \in \mathcal{P}(d,n)$, their convex combination $p \mathbf{E} + (1-p) \mathbf{F}$ is the POVM with the $i$th element given by $p E_i + (1-p) F_i$. A projective measurement is a POVM whose effects are orthogonal projectors. We denote by $\mathbb{P}(d,n)$ the set of projective POVM, which is a subset of $\mathcal{P}(d,n)$.

There are two techniques to simulate quantum measurements. The first is to manipulate quantum measurements randomly by considering a convex combination of POVMs $\{ \mathbf{N}_k \}$, $\mathbf{N}=\sum_k p_k \mathbf{N}_k$. Classical post-processing of a POVM $\mathbf{N}$ is a strategy in which, upon obtaining an output $j$, one produces the final output $i$ with probability $q(i|j)$. For a given post-processing strategy $q(i|j)$, this procedure gives a POVM $\mathcal{Q}(\mathbf{N})$ with elements given by $[\mathcal{Q}(\mathbf{N})]_i := \sum_j q(i|j)N_j$. We say that a POVM $\mathbf{M} \in \mathcal{P}(d,n)$ is PM-simulable if and only if it can be realized by classical randomization followed by classical postprocessing of some projective measurements.

We say $\{M_i\}_{i \in [N] }$ are 2-value POVM operators if $\{M_i,I-M_i\}$ is a valid 2-value POVM for $i \in [N]$, i.e., $0 \leq M_i \leq I$. For a 2-value POVM operator $M$, we say the corresponding 2-value POVM $\{M,I-M\}$ on some quantum state succeeds if the measurement outcome corresponds to $M$. 

\subsection{Classes of measurements on bipartite states}\label{sec:measurement}

Let us formally introduce some specific classes of 2-value POVMs acting on bipartite states. See Appendix C in \cite{HM13} for a reference. Each class of measurement operators describes operators on $\mathbb{C}^d \otimes \mathbb{C}^d$, and for a 2-value POVM $\{M, I-M\}$, $M$ corresponds to ``accept'' and $I-M$ corresponds to ``reject''.

\begin{definition}
    We say a 2-value POVM $\{M, I-M\} \in \mathrm{BELL}$ if $M$ can be expressed as 
    \[
        M = \sum_{(i,j)\in S} \alpha_i \otimes \beta_j,
    \]
    where $\sum_i \alpha_i = I$ and $\sum_j \beta_j = I$, and $S$ is a set of pairs of indices.
\end{definition}

In other words, with BELL measurements, systems are measured locally, and the verifier accepts based on the measurement results.

Let us next introduce some classes where adaptive measurements are allowed. LOCC is the abbreviation for local operations and classical communication.
\begin{definition}
     We say a 2-value POVM $\{M, I-M\} \in \mathrm{LOCC}_1$ if $M$ can be realized by measuring the first system and then choosing a measurement on the second system conditional on the result of the first measurement.  Such $M$ can be written as
     \[
        M = \sum_i \alpha_i \otimes M_i,
     \]
    where $\sum_i \alpha_i = I$ and $0\leq M_i \leq I$ for each $i$.
\end{definition}

\begin{definition}
     We say a 2-value POVM $\{M, I-M\} \in \mathrm{LOCC}$ is if $M$ that can be realized by alternating partial measurements on the two systems a finite number of times, choosing each measurement conditioned on the previous outcomes. An inductive definition is that M is in $\mathrm{LOCC}$ if there exist operators $\{E_i\}, \{M_i\}$, with $\sum_i E_i\leq I$ and each $M_i\in{\rm LOCC}$, such that either $M = \sum_i (\sqrt{E_i}\otimes I)M_i(\sqrt{E_i}\otimes I)$ or $M = \sum_i (I\otimes \sqrt{E_i})M_i(I\otimes \sqrt{E_i})$.   For the base case, it suffices to take $I\in {\rm LOCC}$.
\end{definition}

We finally introduce a set of operators that can be written as a sum of separable operators.

\begin{definition}
     We say a 2-value POVM $\{M, I-M\} \in \mathrm{SEP}$ if
    \begin{equation*}
        M = \sum_i \alpha_i \otimes \beta_i
    \end{equation*}
    for some positive semidefinite matrices $\{\alpha_i\},\{\beta_i\}$.
\end{definition}

\begin{definition}
    We say a 2-value POVM $\{M, I-M\} \in \mathrm{SEP\text{-}BOTH}$ if $M$ and $I-M$ both have the form $\sum_i\alpha_i\otimes\beta_i$.
\end{definition}

\begin{definition}
    $\mathrm{ALL}$ is any $M$ such that $0 \leq M \leq I$.
\end{definition}

Note that SEP-BOTH is a natural relaxation of LOCC
because they preserve the property that both $M$ and $I-M$ must be
realizable through local operations and classical communication.  On
the other hand, SEP is more natural when we consider $M$ by
itself and do not wish to consider additional constraints on $I-M$.

These sets satisfy the following inclusions

$$
\begin{array}{ccccccccccc}
{\rm BELL} &\subset& {\rm LOCC}_1 &\subset& {\rm LOCC}& \subset &
\text{SEP-BOTH} & \subset & {\rm SEP} & \subset & {\rm ALL}.
\end{array}
$$

\subsection{Quantum LOCC SMP protocols}

In this section, let us present the definitions of quantum one-way-LOCC and two-way-LOCC SMP protocols more explicitly. For clarity, let us decompose the referee into the two referees $\rm{Ref}_A$ and $\rm{Ref}_B$ who receive quantum messages from Alice and Bob respectively. In the definitions, we assume that $\rm{Ref}_A$ conducts the first measurement.

\begin{definition}[Quantum one-way-LOCC SMP protocols, $\Q^{||,\mathrm{LOCC}_1}(R)$]
    For a relation $R \subseteq X \times Y \times Z$, Alice has a part of the input $x \in X$ and Bob has the other part of the input $y \in Y$. Alice sends a quantum state $\rho_x$ to $\rm{Ref}_A$ and Bob sends a quantum state $\sigma_y$ to $\rm{Ref}_B$. $\rm{Ref}_A$ conducts some POVM measurement on $\rho_x$ and sends a measurement result to $\rm{Ref}_B$. $\rm{Ref}_B$ conducts some POVM measurement on $\sigma_y$ based on the classical message from $\rm{Ref}_A$. The measurement result of $\rm{Ref}_B$ is the output of the whole communication protocol. We call this protocol a quantum one-way-LOCC protocol for $R$, and say that a quantum one-way-LOCC SMP protocol computes a relation $R \subseteq X \times Y \times Z$, if for all inputs $(x,y) \in X \times Y$, $\rm{Ref}_B$ computes $z \in Z$ such that $(x,y,z) \in R$ with high probability (say $\frac{2}{3}$). The complexity $\Q^{||,\mathrm{LOCC}_1}(R)$ is the sum of sizes of the two messages from Alice and Bob, i.e., $|\rho_x|+|\sigma_y|$.
\end{definition}

\begin{definition}[Quantum two-way-LOCC SMP protocols, $\Q^{||,\mathrm{LOCC}}(R)$]
    For a relation $R \subseteq X \times Y \times Z$, Alice has a part of the input $x \in X$ and Bob has the other part of the input $y \in Y$. Alice sends a quantum state $\rho_x$ to $\rm{Ref}_A$ and Bob sends a quantum state $\sigma_y$ to $\rm{Ref}_B$. $\rm{Ref}_A$ measures $\rho_x$ and sends a measurement result to $\rm{Ref}_B$. Then, $\rm{Ref}_B$ measures $\sigma_y$ based on the measurement result of $\rho_x$ and sends a measurement result to $\rm{Ref}_A$. Then, $\rm{Ref}_A$ measures the residue state of $\rho_x$ based on all the previous measurement results. After iterating these measurements and communication, $\rm{Ref}_A$ or $\rm{Ref}_B$ finally outputs $z \in Z$ based on all the previous measurement outcomes. We call this protocol a quantum two-way-LOCC protocol for $R$, and say that a quantum two-way-LOCC SMP protocol computes a relation $R \subseteq X \times Y \times Z$, if for all inputs $(x,y) \in X \times Y$,  $\rm{Ref}_A$ or $\rm{Ref}_B$ outputs $z \in Z$ such that $(x,y,z) \in R$ with high probability (say $\frac{2}{3}$). The complexity $\Q^{||,\mathrm{LOCC}}(R)$ is the sum of two messages from Alice and Bob, i.e., $|\rho_x|+|\sigma_y|$.
\end{definition}

\subsection{Quantum incoherent one-way communication protocols}

In this paper, we introduce a new quantum one-way communication protocol. There is a quantum channel from Alice to Bob, and Bob's quantum computer can indeed receive quantum messages from Alice. It can perform some complex measurement on the message state, but the measurement cannot depend on his input and it cannot keep residue states after the first measurement. Let us call complexity in this scenario quantum incoherent one-way communication complexity, and denote it by $\Q^{1,\perp}(R)$.

\begin{definition}[Quantum incoherent one-way communication complexity, $\Q^{1,\perp}(R)$]\label{def:Q_one-way_incoherent}
    Given a relation $R \subseteq X \times Y \times Z$, Alice has a part of the input $x \in X$ and Bob has the other part of the input $y \in Y$. Alice produces a quantum state $\rho_x$ and sends it to Bob. Then, Bob performs some POVM $\{M_m\}_m$ on $\rho_x$ and then outputs $z \in Z$ based on $m$ and $y$. We say that a quantum one-way communication protocol computes a relation $R \subseteq X \times Y \times Z$, if for all inputs $(x,y) \in X \times Y$, Bob computes $z \in Z$ such that $(x,y,z) \in R$ with high probability (say $\frac{2}{3}$). The cost of the protocol is the size of the message $\rho_x$ and we denote by $\Q^{1,\perp}(R)$ the minimum cost for all such communication protocols to compute $R$.
\end{definition}

Note that $\Q^{1,\perp}(R) \leq \R^1(R)$ because classical messages can be encoded in quantum messages with the computational basis and Bob measures it with the basis.
\section{Lower bound against one-way LOCC measurement}\label{sec:one-way_locc}

\subsection{Proof}\label{subsec:proof_1waylocc}

We first show that a variant of Newman's Theorem holds in our setting. We assume that in the hybrid SMP model, Alice sends a classical message and Bob sends a quantum message.
Let us denote by $\R\Q^{||,pub}(F)$ and $\R\Q^{||,pub}(R)$ the complexity in the hybrid SMP model of a function $F$ and a relation $R$, respectively, when Alice and the referee share unlimited randomness.

\begin{lemma}\label{lem:newman}
    Let $R \subseteq X \times Y \times Z$ be any relational problem. For every $\epsilon>0$ and $\delta>0$, $\R\Q_{\epsilon+\delta}^{||}(R) \leq \R\Q_\epsilon^{||,pub} (R) + O(\log (\log |X| + \log |Y|) + \log (\frac{1}{\delta}))$.
\end{lemma}

\begin{proof}
    We will show that any hybrid SMP protocol $\mathcal{P}$ with unlimited random bits between Alice and the referee can be transformed into another hybrid SMP protocol $\mathcal{P}'$ in which Alice and the referee share only $O(\log n + \log (\frac{1}{\delta}))$ random bits while increasing the error by only $\delta$. Since the amount of randomness is small, by making Alice send all the random bits to the referee, we have the desired protocol. Let $\Pi$ be the probabilistic distribution of the shared randomness of the protocol $\mathcal{P}$.

    Let $V(x, y, r)$ be a random variable which is defined as the probability that $\mathcal{P}$’s output $z \in Z$ on input $(x, y) \in X \times Y$ and random string $r$ shared by Alice and the referee satisfy $(x, y, z) \notin R$. Because $\mathcal{P}$ computes $R$ with $\epsilon$ error, we have $\mathbb{E}_{r \in \Pi} [V(x,y,r)] \leq \epsilon$ for all $(x,y)$. We will build a new protocol that uses fewer random bits, using the probabilistic method. Let $t$ be a parameter, and let $r_1,\ldots,r_t$ be $t$ strings. For such strings, let us define a protocol $\mathcal{P}_{r_1,\ldots,r_t}$ as follows: Alice and the referee choose $1 \leq i \leq t$ uniformly at random and then proceed as in $\mathcal{P}$ with $r_i$ as their common random string. We now show that there exist strings $r_1,\ldots,r_t$ such that $\frac{1}{t} \sum_{i=1}^t [V(x,y,r_i)] \leq \epsilon + \delta$ for all $(x,y)$. For this choice of strings, the protocol $\mathcal{P}_{r_1,\ldots,r_t}$ is the desired protocol.

    To do so, we choose the $t$ values $r_1,\ldots,r_t$ by sampling the distribution $\Pi$ $t$ times. Consider a particular input pair $(x,y)$ and compute the probability that $\frac{1}{t} \sum_{i=1}^t V(x,y,r_i) > \epsilon + \delta$. By Hoeffding's bound (\cref{fact:hoeffding}), since $\mathbb{E}_{r \in \Pi}[V(x,y,r)] \leq \epsilon$, we get
    \[
        \mathrm{Pr} \left[ \left( \frac{1}{t}  \sum_{i=1}^t V(x,y,r_i) - \epsilon \right) > \delta \right] \leq 2e^{-2\delta^2 t}.
    \]
    By choosing $t=O(\frac{ \log |X| + \log |Y| }{\delta^2})$, this is smaller than $|X|^{-1} |Y|^{-1}$. Thus, for a random choice of $r_1,\ldots,r_t$, the probability that for some input $(x,y)$, $\frac{1}{t} \sum_{i=1}^t [V(x,y,r_i)] > \epsilon + \delta$ is smaller than $|X||Y| |X|^{-1} |Y|^{-1} = 1$. This implies that there exists a choice of $r_1,\ldots,r_t$ where for every $(x,y)$ the error of the protocol $\mathcal{P}_{r_1,\ldots,r_t}$ is at most $\epsilon + \delta$. Finally, note that the number of random bits used by the protocol $\mathcal{P}_{r_1,\ldots,r_t}$ is $\log t = O(\log (\log |X| + \log |Y|) + \log (\frac{1}{\delta}))$.
\end{proof}

We also need the result by Oszmaniec, Guerini, Wittek, and Ac{\'i}n \cite{OGWA17}, who showed that any POVM can be simulated by projective measurements with randomization and postprocessing.

\begin{lemma}[Theorem 1 in \cite{OGWA17}]\label{lem:OGWA17}
    Let $\mathbf{S}\mathbb{P}(d^2,nd)$ be the set of PM simulable, $n$-outcome POVMs on $\mathcal{D}(\mathbb{C}^d \otimes \mathbb{C}^d)$. Let $\mathbf{M} \in \mathcal{P}(d,n)$ be an arbitrary $n$-outcome POVM on $\mathcal{D}(\mathbb{C}^d)$ and $\ket{\phi}$ be some fixed pure state on $\mathcal{B}(\mathbb{C}^d)$. Then there exists a PM-simulable POVM $\mathbf{N} \in \mathbf{S}\mathbb{P}(d^2,nd)$ such that $\mathrm{tr}(\rho M_i) = \mathrm{tr}((\rho \otimes \ket{\phi} \bra{\phi}) N_i)$ for $i=1,\ldots,n$ \footnote{Since the ancilla state is fixed, $\mathrm{tr}((\rho \otimes \ket{\phi} \bra{\phi}) N_i)=0$ for $i=n+1,\ldots,nd$. The number of the outcomes of the POVM $\mathbf{N}$ is $nd$ to make $\mathbf{N}$ valid (i.e., $\sum_{i=1}^{nd} N_i = I$).} and for all states $\rho \in \mathcal{D}(\mathbb{C}^d)$.
\end{lemma}

We are now ready to prove the main result of this section.

\begin{theorem}\label{thm:one-way_LOCC}
    Let $R \subseteq X \times Y \times Z$ be any relation. Suppose that there exists a quantum one-way-LOCC SMP protocol $\mathcal{P}$ for $R$ whose message sizes are $a$ and $b$, with $a$ being the size of Alice's message, on which the first measurement is performed by $\rm{Ref}_A$. Then, $\R\Q(R) \leq 2a + b + O(\log (\log |X| + \log |Y|))$.
\end{theorem}

\begin{proof}
    We give a new hybrid SMP protocol $\mathcal{P}'$ for $R$ with shared randomness between Alice and the referee constructed from the protocol $\mathcal{P}$. Let $\mathbf{M}$ be the $m$-outcome POVM performed by $\rm{Ref}_A$ and let $\rho_x$ be a quantum message from Alice to $\rm{Ref}_A$ in the protocol $\mathcal{P}$. From \cref{lem:OGWA17}, there exist a fixed pure state $\ket{\phi} \in \mathcal{B}(\mathbb{C}^{2^a})$ and a PM-simulable POVM $\mathbf{N} \in \mathbf{S}\mathbb{P}(2^{2a},m 2^a)$ such that $\mathrm{tr}(\rho_x M_i) = \mathrm{tr}((\rho_x \otimes \ket{\phi} \bra{\phi}) N_i)$ for $i=1,\ldots,m$. Since $\mathbf{N}$ is a PM-simulable POVM, there exist a set of projectors $\{\mathbf{N}_k\}$ and a probability distribution $\{p_k\}_k$ such that $\mathbf{N} = \sum_k p_k \mathbf{N}_k$. The measurement results of each $\mathbf{N}_k$ on $\rho_x \otimes \ket{\phi} \bra{\phi}$ are represented by $2a$ bits. Consider a hybrid SMP protocol $\mathcal{P}'$ in which Alice and the referee share the probability distribution $\{p_k\}_k$ and Alice performs $\sum_k p_k \mathbf{N}_k$ on $\rho_x \otimes \ket{\phi}\bra{\phi}$ and sends the measurement result of $2a$ bits. Then, the referee can recover the original result of $\mathbf{M}$ from the $2a$ bits message from Alice because the referee knows the index $k$ from the shared randomness with Alice, and thus simulates $\rm{Ref}_B$ exactly.
    
    Let us apply \cref{lem:newman} for the protocol $\mathcal{P}'$ by taking $\delta$ as a sufficiently small constant. We then obtain a private-coin hybrid SMP protocol $\mathcal{P}''$ to solve $R$ with bounded error and the complexity of the protocol $\mathcal{P}''$ is $2a + O(\log (\log |X| + \log |Y|)) + b$.
\end{proof}

Since for a Boolean function $F:\{0,1\}^n \times \{0,1\}^n \rightarrow \{0,1\}$, $|X|=|Y| = 2^n$, from the lower bound in the hybrid scheme for $\EQ_n$ (\cref{thm:lower_bound_hybrid}), we have a lower bound in quantum one-way-LOCC SMP protocols for $\EQ_n$.

\begin{corollary}\label{cor:EQ}
    $\Q^{||,\mathrm{LOCC}_1}(\EQ_n) = \Omega(\sqrt{n})$.
\end{corollary}

Since $\mathrm{BELL} \subseteq \mathrm{LOCC_1}$, we also have a lower bound with the BELL measurements.

\begin{corollary}
    $\Q^{||,\mathrm{BELL}}(\EQ_n) = \Omega(\sqrt{n})$.
\end{corollary}

Moreover, we can obtain a more general result for any Boolean function. For any Boolean function $F:\{0,1\}^n \times \{0,1\}^n \rightarrow \{0,1\}$, it is known that $\R^{||}(F) \leq \D^{||}(F)^2$ \cite{BK97}, and $\R\Q^{||}(F) \leq \R^{||}(F)^2$ \cite{GRdW08}. Combining this with \cref{thm:one-way_LOCC}, we have the following claim.

\begin{corollary}\label{cor:quartic}
    For any Boolean function $F:\{0,1\}^n \times \{0,1\}^n \rightarrow \{0,1\}$, if $\D^{||}(F) = \Omega(n^c)$ for a constant $c > 0$, $\Q^{||,\mathrm{LOCC_1}} (F) = \Omega(n^{\frac{c}{4}})$.
\end{corollary}

\subsection{Implication for quantum incoherent one-way communication complexity}

By applying the technique developed in \cref{subsec:proof_1waylocc}, we show that, in the setting of the quantum incoherent one-way communication protocols, quantum messages can be replaced by classical messages with very small overhead.
\begin{corollary}\label{cor:incoherent}
    Let $R \subseteq X \times Y \times Z$ be any relation. $\R^1(R) \leq 2 \Q^{1,\perp}(R) + O(\log (\log |X| + \log |Y|)$.
\end{corollary}
\begin{proof}
    Let $\mathcal{P}$ be an incoherent quantum one-way communication protocol for $R$. Let $a$ be the qubit size from Alice to Bob. Let $\mathbf{M} = \{M_i\}_{i \in [m]}$ be the $m$-outcome POVM performed by Bob and let $\rho_x$ be a quantum message from Alice to Bob in the protocol $\mathcal{P}$. From \cref{lem:OGWA17}, there exist a fixed pure state $\ket{\phi} \in \mathcal{B}(\mathbb{C}^{2^a})$ and a PM-simulable POVM $\mathbf{N} \in \mathbf{S}\mathbb{P}(2^2a,m2^a)$ such that $\mathrm{tr}(\rho_x M_i) = \mathrm{tr}((\rho_x \otimes \ket{\phi} \bra{\phi}) N_i)$ for $i=1,\ldots,m$. Since $\mathbf{N}$ is a PM-simulable POVM, there exist a set of projectors $\{\mathbf{N}_k\}$ and a probability distribution $\{p_k\}_k$ such that $\mathbf{N} = \sum_k p_k \mathbf{N}_k$. The measurement result of each $\mathbf{N}_k$ on $\rho_x \otimes \ket{\phi} \bra{\phi}$ is represented by $2a$ bits. Consider a classical one-way communication protocol $\mathcal{P}'$ in which Alice and Bob share the probability distribution $\{p_k\}_k$ and Alice performs $\sum_k p_k \mathbf{N}_k$ on $\rho_x \otimes \ket{\phi}\bra{\phi}$ and sends the measurement result with $2a$ bits. Then, Bob can recover the original result of $\mathbf{M}$ from the $2a$ bits message from Alice because Bob knows the index $k$ from the shared randomness with Alice.
    We then apply a variant of the Newman's theorem proven by the same way as \cref{lem:newman}, and we obtain a classical one-way communication protocol whose complexity is $2a + O(\log (\log |X| + \log |Y|))$, which concludes the proof.
\end{proof}
\section{Lower bound against two-way LOCC measurement}\label{sec:two-way_locc}

In \cref{subsec:replace}, we first describe how to replace quantum messages by deterministic messages. To illustrate our idea and analysis more clearly, we will give a proof of a lower bound in quantum two-value two-round LOCC SMP protocols as a first step in \cref{subsec:warm-up}. We then prove our lower bound in quantum two-value multiple-round LOCC SMP protocols in \cref{subsec:many_rounds}.

\subsection{Replacing quantum messages by classical messages}\label{subsec:replace}

\begin{lemma}\label{lem:replace}
    Suppose Alice has the classical description of an arbitrary $q$-qubit density matrix $\rho$, and Bob has $2^c$ $2$-value POVM operators $\{E_b\}_{b \in \{0,1\}^c}$. For any $\delta >0$, there exists a deterministic message of $O(\frac{q}{\delta^3} \log (\frac{q}{\delta}) (c+ \log \frac{1}{\delta}) )$ bits from Alice that enables Bob to output values $p'_b$ satisfying that $|p_b-p'_b| \leq \delta$ simultaneously for all $b \in \{0,1\}^c$ where $p_b = \mathrm{tr}(E_b \rho)$.
\end{lemma}

For completeness, we will give a proof of this theorem. See also Theorem 5 in \cite{GRdW08}, Theorem 4 in \cite{Aar18}, and Theorem 6 in \cite{ACH+18}. We will need the quantum union bound for the proof.

\begin{lemma}[Quantum union bound, see, e.g., Lemma 3.1 in \cite{Wil13}]\label{lem:quantum_union_bound}
    Suppose that for a quantum state $\rho$, we conduct a sequence of 2-value POVM operators $\{M_i\}_{i\in[k]}$ such that $\mathrm{tr}(M_i \rho) \geq 1 - \delta$. Then, the probability that all the measurements succeed is at least $1-2\sqrt{k\delta}$.
\end{lemma}

\begin{proof}[Proof of \cref{lem:replace}]
    Suppose that Alice sends $r$ many copies of her state.
    Let $\rho'=\rho^{\otimes r}$ be the state she sends, and $r q$ is the total number of qubits. Define the operator
    \[
        F_b=\frac{1}{r}\sum_{j=1}^r E_b^{(j)},
    \]
    where $E_b^{(j)}$ applies $E_b$ to the $j$th copy. This operator gives the fraction of successes if you separately measure each of the $r$ copies of $\rho$ with $E_b$. Hoeffding's bound (\cref{fact:hoeffding}) implies that the outcome $p_b'$ of this measurement applied to $\rho'$ will probably be close to its expectation $p_b=\mathrm{tr}(E_b\rho)$:
    \begin{equation}\label{eq:use_chernoff}
        \Pr[|p'_b-p_b|>\delta/4]\leq 2 e^{-\frac{\delta^2 r}{8}}.
    \end{equation}

    Let us show what is Alice's classical message. Consider all $b=1,\ldots,2^c$ in order. We will sequentially build $rq$-qubit density matrices $\rho_b$, one for each $E_b$.
    Alice's classical message will enable Bob to reconstruct this entire sequence. Let us say $b$ is \emph{good} if $|\mathrm{tr}(F_b \rho_b)-p_b| \leq \delta$, and say $b$ is \emph{bad} otherwise. Note that if Bob has a classical description of a good $\rho_b$, then he can approximate $p_b$ to within $\pm \delta$ (since he knows what $F_b$ is). We start with the completely mixed state: $\rho_1=\frac{I}{2^{rq}}$ and define the subsequent $\rho_b$ one by one, as follows. If $b$ is good, then define $\rho_{b+1}$ as equal to $\rho_b$. If $b$ is bad, Alice adds the pair $(b,\widetilde{p}_b)$ to her message, where $\widetilde{p}_b$ is the $\log(1/\delta)+O(1)$ most significant bits of $p_b$, and then $|\widetilde{p}_b-p_b| \ll \delta$. In this case, let $M_b$ be the projector on the subspace spanned by the eigenvectors of $F_b$ with eigenvalues in the interval $[\widetilde{p}_b-\delta/2,\widetilde{p}_b+\delta/2]$, and let $\rho_{b+1}$ be the renormalized projection of $\rho_b$ on this subspace. Note that we will later show $\mathrm{tr}(M_b \rho_b)$ is nonzero in \cref{eq:lower_bound} and thus this renormalized projection is well defined. Continuing all the way to $b=2^c$, we obtain a message $(b_1,\widetilde{p}_{b_1}),\ldots,(b_T,\widetilde{p}_{b_t})$ for some $t$. We need to show two things: (i) this message enables Bob to approximate all $p_b$ to within $\pm \delta$, and (ii) $t=O(rq)$, which implies that the message length is 
    \begin{equation}\label{eq:total_length}
        O\left(t \left( c+\log \frac{1}{\delta} \right) \right)
    \end{equation}    
    bits. We will show these two things in turn.

    First, we will show (i): the construction of the messages works. Note that Bob knows which $b\in[2^c]$ are bad, since those $b$ are exactly the ones in Alice's message. Bob can in fact compute the whole sequence $\rho_1,\ldots,\rho_{2^c}$ given the message: $\rho_1=\frac{I}{2^{rq}}$; if $b$ is good, then $\rho_{b+1}=\rho_b$; if $b$ is bad, then $(b,\widetilde{p}_b)$ is part of Alice's message and $\rho_{b+1}$ can be computed from this information. Suppose that Bob wants to approximate $p_b=\mathrm{tr}(E_b\rho)$. If $b$ is good then by definition $|\mathrm{tr}(F_b\rho_b)-p_b| \leq \delta$ and Bob can calculate $\mathrm{tr}(F_b\rho_b)$. If $b$ is bad, then the pair $(b,\widetilde{p}_b)$ is part of Alice's message, so Bob knows $p_b$ with sufficient precision. Hence Bob can approximate all $p_b$ up to $\pm \delta$, for all $b$ simultaneously.

    Second, we will show (ii): $t=O(rq)$. Define $\eta=1-\delta/4$ and assume $t \leq (\frac{q}{\delta})^{10}$ (this assumption will be justified later). We consider the sequence $b_1,\ldots,b_t$ of the first $t$ bad $b$'s.
    Let
    \[
        p=\mathrm{tr} \left( M_{b_t} \cdots M_{b_1} \frac{I}{2^{rq}} M_{b_1} \cdots M_{b_t} \right)
    \]
    be the probability that all $t$ measurements succeed if we start with the completely mixed state and sequentially measure $M_{b_1},\ldots,M_{b_t}$. 

    We will upper bound and lower bound $p$ and do the upper bound first. If we sequentially measure $M_{b_1},\ldots,M_{b_t}$, starting from the completely mixed state, and if all $t$ measurements succeed, then we exactly have the sequence of density matrices $\rho_{b_1}= \frac{I}{2^{rq}},\ldots,\rho_{b_t},\rho_{b_{t+1}}$. 
    We will show the claim that if $\rho_b$ is bad, then $\mathrm{tr}(M_b\rho_b)\leq \eta$. Let $X$ denote the random variable representing the outcome of measuring $\rho_b$ with the observable $F_b$. Note that $X$ takes values in $[0,1]$. Assume $\mathrm{tr}(M_b\rho_b)=\Pr[|X-\widetilde{p}_b|\leq\delta/2] > \eta$. Let us evaluate the value $\mathrm{tr}(F_b\rho_b) = \E[X]$. From the assumption, we have
    \[
        \Pr [ X \geq \widetilde{p}_b - \delta/2 ] \geq \Pr[|X-\widetilde{p}_b|\leq\delta/2] > \eta.
    \]
    Then, from Markov's inequality (\cref{fact:markov}), if $\widetilde{p}_b - \delta/2 >0$, we have
    \[
        \E[X] > \eta (\widetilde{p}_b - \delta/2).
    \]
    Otherwise, it is trivial that $\E[X] \geq 0 \geq \eta (\widetilde{p}_b - \delta/2)$. Since $X$ takes values in $[0,1]$, by a similar discussion, we have
    \[
        \E[X] \leq \eta (\widetilde{p}_b + \delta/2) + 1 - \eta.
    \]
    Since $|\widetilde{p}_b-p_b| \ll \delta$ and thus $\widetilde{p}_b \leq 1 + \delta/100$, $\mathrm{tr}(F_b\rho_b) = \E[X]$ must necessarily be in the range
    \[
        [\eta (\widetilde{p}_b - \delta/2), \eta (\widetilde{p}_b + \delta/2) + 1-\eta] \subseteq [\widetilde{p}_b - 3\delta/4, \widetilde{p}_b + 3\delta/4] \subseteq [p_b - \delta, p_b + \delta],
    \]
    and hence $\rho_b$ is good. Thus we have the claim by contraposition, and then the probability that all $t$ measurements succeed is 
    \begin{equation}\label{eq:upper_bound}
        p\leq \eta^t.
    \end{equation}

    Now we lower bound on $p$. Note that $M_b$ succeeds on $\rho'$ if and only if the outcome $p'_b$ of the observable $F_b$ is at most $\delta/2$ away from the number $\widetilde{p}_b$, which is the truncated version of $p_b=\mathrm{tr}(E_b\rho)$ (recall $|\widetilde{p}_b-p_b|\ll \delta$). Hence by \cref{eq:use_chernoff}, we have
    \begin{equation}\label{eq:nonzero}
        \mathrm{tr}(M_b\rho')=\Pr[|p'_b-\widetilde{p}_b|\leq\delta/2]\geq \Pr[|p'_b-p_b|\leq\delta/4]\geq 1-2 e^{-\frac{\delta^2 r}{8}}.
    \end{equation}
    This allows us to measure $\rho'$ with $M_b$ while disturbing the state by only an insignificant amount. If we measure each $M_b$, for the first $t$ bad $b$'s in sequence, starting from $\rho'$, then from \cref{lem:quantum_union_bound} with probability at least 
    \[
        1-2\sqrt{2t e^{-\frac{\delta^2 r}{8}}}
    \]
    all measurements will succeed. Set $r:=\frac{C}{\delta^2}\log\frac{q}{\delta}$ for a sufficiently large constant $C$. Then we have
    \[
        1-2\sqrt{2t e^{-\frac{\delta^2 r}{8}}} = 1-2\sqrt{t \frac{\delta^{10}}{400 q^{10}}} \geq \frac{1}{2},
    \]
    where we use the assumption $t \leq (\frac{q}{\delta})^{10}$. Moreover, the completely mixed state can be written as $\frac{I}{2^{rq}}=\frac{1}{2^{rq}}\rho'+(1-\frac{1}{2^{rq}})\rho''$ where $\rho''$ is orthogonal to $\rho'$. Hence, if we start from $\frac{I}{2^{rq}}$, then the probability of all measurements succeeding is 
    \begin{equation}\label{eq:lower_bound}
        p\geq \frac{1}{2^{rq+1}}.
    \end{equation}

    Combining the upper bound \cref{eq:upper_bound} and lower bound \cref{eq:lower_bound}, we have
    \[
        \frac{1}{2^{rq+1}} \leq \left(1-\frac{\delta}{4} \right)^t.
    \]
    Since $rq = O(\frac{q}{\delta^2} \log \frac{q}{\delta})$, we have
    \[
        t = O\left(\frac{q}{\delta^3} \log \frac{q}{\delta}\right).
    \]
    This bound also justifies the assumption $t \leq (\frac{q}{\delta})^{10}$. Therefore, from \cref{eq:total_length}, the total length of the deterministic message is
    \[
        O\left(\frac{q}{\delta^3} \log \left(\frac{q}{\delta}\right) \left(c+ \log \frac{1}{\delta} \right) \right)
    \]
    bits, as claimed.
\end{proof}

\subsection{Warm-up case: 2-round-LOCC SMP protocols}\label{subsec:warm-up}

\begin{proposition}\label{prop:warm-up}
    Assume that there exists a 2-value 2-round LOCC SMP protocol $\mathcal{P}$ to solve $\EQ_n$ with high probability. Then the number of qubits of Alice's message
    is $\Omega(n/\log n)$.
\end{proposition}

\begin{proof}
    Let $\{\rho_x\}_{x\in \{0,1\}^n}$ be a quantum encoding by Alice of her input $x$, and $\{\sigma_y\}_{y \in \{0,1\}^n}$ be a quantum encoding by Bob of his input $y$.
    Let $r$ be a number of rounds, and $a \in \{0,1\}^{r}$ be previous measurement results by $\rm{Ref}_A$, and $b \in \{0,1\}^{r}$ be previous measurement results by $\rm{Ref}_B$. Let $h$ be all previous measurement results $a_0b_0\cdots a_{r-1}b_{r-1}$. For $m \in \{0,1\}$, let $M_{m|h}$ be a measurement operator of the $(r+1)$th measurement by $\rm{Ref}_A$, for previous measurement results $h$. Without loss of generality, we assume that the first measurement is done by $\rm{Ref}_A$ for $\rho_x$, and let $M_0$ and $M_1$ be the first measurement operators. We also assume, without loss of generality, that the final measurement outcome by $\rm{Ref}_A$ is the output of the protocol because it depends on all the previous measurement outcomes.
    
    For each $r \in \mathbb{N}$, let us denote by $p^A_{m|h}$ the probability that, conditioned on all previous measurement outcomes $h = a_0 b_0\cdots a_{r-1} b_{r-1}$ which consists of measurement outcomes $a \in \{0,1\}^r$ by $\rm{Ref}_A$ and measurement outcomes $b \in \{0,1\}^r$ by $\rm{Ref}_B$, the $(r+1)$th measurement result by $\rm{Ref}_A$ is $m \in \{0,1\}$. Let us denote by $p^B_{m|h}$ the probability that, conditioned on all previous measurement outcomes $h = a_0 b_0\cdots a_{r-1} b_{r-1} a_r$ which consists of measurement outcomes $a \in \{0,1\}^{r+1}$ by $\rm{Ref}_A$ and measurement outcomes $b \in \{0,1\}^r$ by $\rm{Ref}_B$, the $(r+1)$th measurement result by $\rm{Ref}_A$ is $m \in \{0,1\}$. Let us abbreviate $p^A_{m|\emptyset}$ as $p^A_m$. 
     
    Let us also denote by $v_0 = \mathrm{tr}(M_0^\dagger M_0 \rho_x)$, $v_1 = \mathrm{tr}(M_1^\dagger M_1 \rho_x) = 1 - v_0$,
    \begin{align*}
    v_{1|00}=\mathrm{tr}(M^{ \dagger}_0 M_{1|00}^{\dagger} M_{1|00} M_0 \rho_x), \quad &  v_{1|01}=\mathrm{tr}(M_0^\dagger M_{1|01}^{\dagger} M_{1|01} M_0 \rho_x), \\
    v_{1|10}=\mathrm{tr}(M^{\dagger}_1 M_{1|10}^{\dagger} M_{1|10} M_1 \rho_x), \quad & v_{1|11}=\mathrm{tr}(M^{\dagger}_1 M_{1|11}^{\dagger} M_{1|11} M_1 \rho_x).
    \end{align*}
    By definition,
    \[
        p^A_0 = v_0,\ \ p^A_1 = v_1.
    \]
    For example, when the result of the first measurement by $\rm{Ref}_A$ is $1$ and that of the first measurement by $\rm{Ref}_B$ is $0$, the probability that the second measurement result by $\rm{Ref}_A$ is $1$ is 
    \begin{equation}\label{eq:fraction}
        p^A_{1|10} = \mathrm{tr}(M_{1|10}^\dagger M_{1|10} \frac{M_1 \rho_x M^\dagger_1}{\mathrm{tr}(M^\dagger_1 M_1 \rho_x)}) = \frac{\mathrm{tr}(M^\dagger_1 M_{1|10}^\dagger M_{1|10} M_1 \rho_x) }{\mathrm{tr}(M^\dagger_1 M_1 \rho_x)} = \frac{v_{1|10}}{v_1},
    \end{equation}
    where we use the cyclic property of the trace. By the same argument, we have 
     \begin{equation}
         p^A_{1|00} = \frac{v_{1|00}}{v_0},\ \ p^A_{1|01} = \frac{v_{1|01}}{v_0},\ \ p^A_{1|10} = \frac{v_{1|10}}{v_1},\ \ p^A_{1|11} = \frac{v_{1|11}}{v_1}.
     \end{equation}
    
    Since $ p^A_{0|00} +  p^A_{1|00} = 1$, we have
    \[
        \frac{v_{0|00}}{v_0} + \frac{v_{1|00}}{v_0} = 1,
    \]
    which implies 
    \[
         v_0 = v_{0|00} + v_{1|00}.
    \]
    By the same argument, we also have
    \[
         v_0 = v_{0|01} + v_{1|01},\ \ v_1 = v_{0|10} + v_{1|10},\ \ v_1 = v_{0|11} + v_{1|11}.
    \]
    
    Then, the probability that the original LOCC protocol accepts is
    \begin{align*}
        & p^A_{0} \cdot p_{0|0}^B \cdot p^A_{1|00} + p^A_{0} \cdot p_{1|0}^B \cdot p^A_{1|01} + p^A_{1} \cdot p_{0|1}^B \cdot p^A_{1|10} + p^A_{1} \cdot p_{1|1}^B \cdot p^A_{1|11} \\
        &= v_{0} \cdot p_{0|0}^B \cdot \frac{v_{1|00}}{v_0} + v_{0} \cdot p_{1|0}^B \cdot \frac{v_{1|01}}{v_0} + v_{1} \cdot p_{0|1}^B \cdot \frac{v_{1|10}}{v_1} +  v_{1} \cdot p_{1|1}^B \cdot \frac{v_{1|11}}{v_1}\\
        &= p_{0|0}^B \cdot v_{1|00} + p_{1|0}^B \cdot v_{1|01} + p_{0|1}^B \cdot v_{1|10} + p_{1|1}^B \cdot v_{1|11}.
    \end{align*}
    
    Next, we will replace quantum messages $\rho_x$ with classical messages $s_x$ using \cref{lem:replace}. Let $\delta$ be a sufficiently small constant, say $\frac{1}{10^6}$. Let us denote by $q$ the number of qubits of the message $\rho_x$. The operators we will consider are 
    \begin{align*}
        \{E_b\} = \{ M_0^\dagger M_0, M_0^\dagger M_{1|00}^\dagger M_{1|00} M_0, M_0^\dagger M_{1|01}^\dagger M_{1|01} M_0, M_1^\dagger M_{1|10}^\dagger M_{1|10} M_1,  M_1^\dagger M_{1|11}^\dagger M_{1|11} M_1 \}.
    \end{align*}
    From the definition, these operators $\{E_b\}$ satisfy $0 \leq E_b \leq I$.
    Since the number of operators we need to care is 5, the length of the classical string $s_x$ is $O(q \log q)$ bits. Let us consider a hybrid SMP scheme that uses $s_x$ and $\sigma_y$ as two messages where the referee simulates the original LOCC interactions. Let us denote by $v'_0, v'_{1|00}, v'_{1|01}, v'_{1|01}, v'_{1|11}$ the five values the referee guesses using $s_x$. We also define another set of variables $v''_0, \ldots, v''_{1|11}$ that the referee will use.
    \begin{itemize}
\item    If $v'_0 > 1$, $v''_0 = 1$. If $v'_0 < 0$, $v''_0 = 0$. Otherwise, $v''_0 = v'_0$. Let $v''_1 = 1-v''_0$.
    
\item    If $v'_{1|00} > v''_0$, $v''_{1|00} = v''_0$. If $v'_{1|00} < 0$, $v''_{1|00} = 0$. Otherwise, $v''_{1|00} = v'_{1|00}$.
\item    If $v'_{1|01} > v''_0$, $v''_{1|01} = v''_0$. If $v'_{1|01} < 0$, $v''_{1|01} = 0$. Otherwise, $v''_{1|01} = v'_{1|01}$.
\item    If $v'_{1|10} > v''_1$, $v''_{1|10} = v''_1$. If $v'_{1|10} < 0$, $v''_{1|10} = 0$. Otherwise, $v''_{1|10} = v'_{1|10}$.
\item   If $v'_{1|11} > v''_1$, $v''_{1|11} = v''_1$. If $v'_{1|11} < 0$, $v''_{1|11} = 0$. Otherwise, $v''_{1|11} = v'_{1|11}$.
    \end{itemize}
    This adjustment is required to make $v''_{0}$, $v''_{1}$, $\frac{v''_{1|00}}{v''_{0}}$, $\frac{v''_{1|01}}{v''_{0}}$, $\frac{v''_{1|10}}{v''_{1}}$, $\frac{v''_{1|11}}{v''_{1}}$ valid probabilities. 
    
    The referee first simulates the first measurement result by $\rm{Ref}_A$. With probability $v''_0$, the referee simulates the result to be $0$ and with probability $v''_1$ the result to be $1$. Next, the referee measures Bob's message based on the value of the first measurement result obtained by the simulation. Finally, the referee simulates the final measurement result by $\rm{Ref}_A$ using $v''_{1|00},v''_{1|01},v''_{1|10},v''_{1|11}$.
    The acceptance probability of the simulation is
    \begin{align*}
        &v''_{0} \cdot p_{0|0}^B \cdot \frac{v''_{1|00}}{v''_{0}} + v''_{0} \cdot p_{0|0}^B \cdot \frac{v''_{1|01}}{v''_{0}} + v''_{1} \cdot p_{1|0}^B \cdot \frac{v''_{1|10}}{v''_{1}} + v''_{1} \cdot p_{1|1}^B \cdot \frac{v''_{1|11}}{v''_{1}} \\
        &= p_{0|0}^B \cdot v''_{1|00} + p_{1|0}^B \cdot v''_{1|01} + p_{0|1}^B \cdot v''_{1|10} + p_{1|1}^B \cdot v''_{1|11}.
    \end{align*}
    If $v''_0=0$, $v''_{1|00} = 0$ and $v''_{1|01} = 0$. If $v''_1=0$, $v''_{1|10} = 0$ and $v''_{1|11} = 0$. Therefore, the quantity does not lose generality.

    If $v'_0 > 1$, $- \delta \leq v_0 - v'_0 \leq v_0 - v''_0 = v_0 - 1 \leq 0$, which implies $|v_0 - v''_0| \leq \delta$. If $v'_0 < 0$, $0 \leq v_0 - v''_0 = v_0 \leq v_0 - v'_0 \leq \delta$, which implies $|v_0 - v''_0| \leq \delta$. If $v''_0 = v'_0$, $|v_0 - v''_0| = |v_0 - v'_0| \leq \delta$.
    If $v'_0 > 1$, $|v_1 - v''_1| = v_1 = 1-v_0 \leq \delta$. If $v'_0 < 0$, $|v_1 - v''_1| = |v_1 - 1| = |v_0| \leq \delta$. If $v''_0 = v'_0$, $|v_1 - v''_1| = |(1-v_0) - (1-v''_0)| \leq \delta$. In summary, we have $|v_0 - v''_0| \leq \delta$ and $|v_1 - v''_1| \leq \delta$ in all cases.
    
    Let us show the quantity $|v_{1|00} - v''_{1|00}|$ is small.
    If $v'_{1|00} > v''_0$, $v''_{1|00} = v''_0$ and we have 
    \[
         - \delta \leq v_{1|00} - v'_{1|00} \leq v_{1|00} - v''_{0} \leq v_{0} - v''_{0} \leq \delta
    \]
    by the definitions, and it implies $|v_{1|00} - v''_{1|00}| = |v_{1|00} - v''_0| \leq \delta$. 
    If $v'_{1|00} < 0$, $v''_{1|00} = 0$ and we have
    \[
        |v_{1|00} - v''_{1|00}| =  v_{1|00} \leq |v_{1|00} - v'_{1|00}| \leq \delta
    \]
    by the definitions.
    Otherwise, $v''_{1|00} = v'_{1|00}$ and $|v_{1|00} - v''_{1|00}| = |v_{1|00} - v'_{1|00}| \leq \delta$. In summary, in all cases, we have $|v_{1|00} - v''_{1|00}| \leq \delta$. By similar discussions, we also have $|v_{1|01} - v''_{1|01}| \leq \delta$.

    Let us next show the quantity $|v_{1|10} - v''_{1|10}|$ is small. If $v'_{1|10} > v''_1$, $v''_{1|10} = v''_1$ and we have
    \[
         - \delta \leq v_{1|10} - v'_{1|10} \leq v_{1|10} - v''_{1} \leq v_{1} - v''_{1} \leq \delta
    \]
    by the definitions, and it implies $|v_{1|10} - v''_{1|10}| = |v_{1|10} - v''_1| \leq \delta$. If $v'_{1|00} < 0$, $v''_{1|00} = 0$ and we have
    \[
        |v_{1|00} - v''_{1|00}| =  v_{1|00} \leq |v_{1|00} - v'_{1|00}| \leq \delta
    \]
    by the definitions.
    If $v'_{1|10} < 0$, $v''_{1|10} = 0$ and we have
    \[
        |v_{1|10} - v''_{1|10}| =  v_{1|10} \leq |v_{1|10} - v'_{1|10}| \leq \delta
    \]
    by the definitions.
    Otherwise, $v''_{1|10} = v'_{1|10}$ and $|v_{1|10} - v''_{1|10}| = |v_{1|10} - v'_{1|10}| \leq \delta$ by the definition. In summary, in all cases, we have $|v_{1|10} - v''_{1|10}| \leq \delta$. By similar arguments, we also have $|v_{1|11} - v''_{1|11}| \leq \delta$.
    
    Therefore, the difference between the true value and the simulation value is
    \begin{align*}
        &|(p_{0|0}^B \cdot v_{1|00} + p_{1|0}^B \cdot v_{1|01} + p_{0|1}^B \cdot v_{1|10} + p_{1|1}^B \cdot v_{1|11}) - (p_{0|0}^B \cdot v''_{1|00} +p_{1|0}^B \cdot v''_{1|10} + p_{0|1}^B \cdot v''_{1|10} + p_{1|1}^B \cdot v''_{1|11})| \\
        & \leq p_{0|0}^B |v_{1|00} - v''_{1|00}| + p_{1|0}^B |v_{1|01} - v''_{1|01}| + p_{0|1}^B |v_{1|10} - v''_{1|10}| + p_{1|1}^B |v_{1|11} - v''_{1|11}| \\
        & \leq (p_{0|0}^B + p^B_{1|0} + p^B_{0|1} + p^B_{1|1}) \delta = 2 \delta,
    \end{align*}
    where we use $p_{0|0}^B + p^B_{1|0} = 1$ and $p^B_{0|1} + p^B_{1|1} = 1$.
    
    If classical messages of hybrid SMP schemes for $\EQ_n$ are deterministic, then the length of the classical message is $\Omega(n)$. This is because, for $2^n$ inputs, there are some conflicts of deterministic messages and the referee cannot distinguish them. We thus have $O(q \log q) = \Omega(n)$ and $q = \Omega(n/\log n)$, which concludes the proof.
\end{proof}

\subsection{Multiple-round LOCC SMP protocols}\label{subsec:many_rounds}

When we increase the number $r$ of rounds by $2$, the number of operators we will care is increased by $4^r$. This is because both the numbers of measurement outcomes by $\rm{Ref}_A$ and $\rm{Ref}_B$ are $2$. Therefore, the total number of the operators we will care for $2r$-round LOCC protocols is 
\begin{equation}\label{eq:num}
    \sum_{i=0}^{r-1} 4^r = \frac{4^r - 1}{4-1} \leq 4^r = 2^{2r}. 
\end{equation}

We first show that each measurement outcome can be simulated by taking ratios of $2^{2r}$ values to simulate quantum $2$-value $2r$-round quantum LOCC SMP protocol. Note that for measurement operators $M_i$ for $i\in[n]$, $0 \leq M_0^\dagger\cdots M_{n-1}^\dagger M_{n-1}\cdots M_0 \leq I$.

\begin{lemma}\label{lem:ratio}
    Let us consider the case where a quantum state $\rho$ is measured $n \in \mathbb{N}$ times, and let $M_i$ be a measurement operator of the $i$th 2-value measurement for $i \in [n]$. Let us denote $v_\emptyset:=1$ and $v_i := \mathrm{tr}(M_0^\dagger\cdots M_{i}^\dagger M_{i}\cdots M_0 \rho )$. Then, the probability that the $n$th measurement succeeds conditioned on the event that all $n-1$ previous measurements succeed, $p_{n-1|n-2}$, is $\frac{v_{n-1}}{v_{n-2}}$.
\end{lemma}
\begin{proof}
    We will show the statement by induction. When $n=1$, $p_{0|\emptyset} = \mathrm{tr}(M_0^\dagger M_0 \rho) = \frac{v_0}{v_\emptyset}$.
    
    For $k>1$, let us consider $p_{k|k-1}$. Let us denote by $\rho_i$ the state after the $(i-1)$th measurement succeeds for $i \in [k]$. In other words, we define
    \[
        \rho_i := \frac{M_i \rho_{i-1} M_i^\dagger}{\mathrm{tr}(M_i^\dagger M_i \rho_{i-1})}
    \]
    for $i \in [k]$. Using the notation, let us assume that 
    \[
        p_{i|i-1} = \mathrm{tr}(M_{i}^\dagger M_{i} \rho_{i-1}) = \frac{v_{i}}{v_{i-1}}
    \]
    for all $i \in [k-1]$.

    Then, we have
    \begin{align*}
        p_{k|k-1} &= \mathrm{tr}(M_k^\dagger M_k \rho_{k-1}) = \frac{\mathrm{tr}(M_k^\dagger M_k M_{k-1} \rho_{k-2} M_{k-1}^\dagger)}{\mathrm{tr}(M_{k-1}^\dagger M_{k-1} \rho_{k-2})} \\
        &= \frac{v_{k-2}}{v_{k-1}} \mathrm{tr}(M_k^\dagger M_k M_{k-1} \rho_{k-2} M_{k-1}^\dagger) = \frac{v_{k-2}}{v_{k-1}} \mathrm{tr}(M_{k-1}^\dagger M_k^\dagger M_k M_{k-1} \rho_{k-2}) \\
        &= \frac{v_{k-2}}{v_{k-1}} \mathrm{tr}(M_{k-1}^\dagger M_k^\dagger M_k M_{k-1} \frac{M_{k-2} \rho_{k-3} M_{k-2}^\dagger}{\mathrm{tr}(M_{k-2}^\dagger M_{k-2} \rho_{k-3})}) \\
        &= \frac{v_{k-2}}{v_{k-1}} \frac{v_{k-3}}{v_{k-2}} \mathrm{tr}(M_{k-2}^\dagger M_{k-1}^\dagger M_k^\dagger M_k M_{k-1} M_{k-2} \rho_{k-3}) \\
        &= \cdots \\
        &= \frac{v_{k-2}}{v_{k-1}} \frac{v_{k-3}}{v_{k-2}} \cdots \frac{v_1}{v_2}\mathrm{tr}(M_2^\dagger \cdot M_k^\dagger M_k \cdots M_2 \rho_1) \\
        &= \frac{v_{k-2}}{v_{k-1}} \frac{v_{k-3}}{v_{k-2}} \cdots \frac{v_1}{v_2} \frac{v_k}{v_1}\\
        &= \frac{v_k}{v_{k-1}}
    \end{align*}
    where we use the cyclic property of the trace.
    Therefore, by induction, we have the claim.
\end{proof}

\begin{lemma}\label{lem:error}
    Let $\mathcal{P}$ be a quantum $2$-value $2r$-round LOCC SMP protocol to solve $\EQ_n$. Using a deterministic message from \cref{lem:replace}, the referee can simulate $\mathcal{P}$ with error $2^r(r+1)\delta$.
\end{lemma}

\begin{proof}
We assume, without loss of generality, that the final measurement outcome by $\rm{Ref}_A$ is the output of the protocol because it depends on all the previous measurement outcomes.
For all previous measurement results $h = a_0 b_0 \cdots$, let us denote by $h[i]$ the $(i+1)$th bit of $h$ and by $h[i,j]$ $(i+1)$th to $(j+1)$th bits of $h$. 
For $h \in \{0,1\}^{2R-1}$, let us define 
\[
v_{m|h} = \mathrm{tr}(M_{h[0]}^\dagger \cdots M_{h[2R-2]|h[0,2R-3]}^\dagger M_{h[2R-2]|h[0,2R-3]} \cdots M_{h[0]} \rho_x).
\] 
We also define $v''_{m|h}$ to make the probabilities valid as in \cref{prop:warm-up} (see the proof of \cref{lem:error} for a formal definition).

The true and simulation probability is the sum of $2^{2r}$ values and, from \cref{lem:ratio}, the total error is
\[
    \left| \sum_{h \in \{0,1\}^{2r}} \prod_{i \in [r]} p^B_{h[2i-1]|h[0,2i-2]} (v_{1|h} - v''_{1|h}) \right| \leq \sum_{h \in \{0,1\}^{2r}} \prod_{i \in [r]} p^B_{h[2i-1]|h[0,2i-2]}  |v_{1|h} - v''_{1|h}|,
\]
where $h[2i-1]$ denotes the $(2i)$th bit of $h$ (i.e., $b_i$) and $h[0,2i-2]$ denotes 1st to $(2i-1)$th bits of $h$ (i.e., $a_0b_0\cdots a_{i-1} b_{i-1} a_i$).

Let us evaluate how much error each term causes.
\begin{lemma}
    $|v_{m|h} - v''_{m|h}| \leq (k+1) \delta$ for all $m \in \{0,1\}$, $h \in \{0,1\}^{2k}$.
\end{lemma}

\begin{proof}
    Let us first show the base case where $k=0$ and $h = \emptyset$. If $v'_0 > 1$, $v''_0 := 1$. If $v'_0 < 0$, $v''_0 := 0$. Otherwise, $v''_0 := v'_0$. Let $v''_1 := 1-v''_0$. 
    
    If $v'_0 > 1$, $- \delta \leq v_0 - v'_0 \leq v_0 - v''_0 = v_0 - 1 \leq 0$, which implies $|v_0 - v''_0| \leq \delta$. If $v'_0 < 0$, $0 \leq v_0 - v''_0 = v_0 \leq v_0 - v'_0 \leq \delta$, which implies $|v_0 - v''_0| \leq \delta$. If $v''_0 = v'_0$, $|v_0 - v''_0| = |v_0 - v'_0| \leq \delta$.
    If $v'_0 > 1$, $|v_1 - v''_1| = v_1 = 1-v_0 \leq \delta$. If $v'_0 < 0$, $|v_1 - v''_1| = |v_1 - 1| = |v_0| \leq \delta$. If $v''_0 = v'_0$, $|v_1 - v''_1| = |(1-v_0) - (1-v''_0)| \leq \delta$. In summary, we have $|v_0 - v''_0| \leq \delta$ and $|v_1 - v''_1| \leq \delta$ in all cases.
    
    Assume that $|v_{m|h} - v''_{m|h}| \leq (t+1) \delta$ for all $m \in \{0,1\}$ and $h \in \{0,1\}^{2t}$. If $v'_{0|hab} > v''_{a|h}$, $v''_{0|hab} = v''_{a|h}$. If $v'_{0|hab} <0$, $v''_{0|hab} = 0$. Otherwise, $v''_{0|hab}= v'_{0|hab}$. Let $v''_{1|hab} = v''_{a|h} - v''_{0|hab}$. We want to prove that $|v_{m|hab} - v''_{m|hab}| \leq (t+2) \delta$ for $m \in \{0,1\}$, $h \in \{0,1\}^t$, $a \in \{0,1\}$ and $b \in \{0,1\}$.

    For conciseness, let us fix $a=0$, $b=0$ and we will prove $|v''_{m|h00} - v_{m|h00}| \leq (t+2) \delta$ for $m \in \{0,1\}$. If $v'_{0|h00} > v''_{0|h}$,
    \[
        -\delta \leq v_{0|h00} - v'_{0|h00} \leq v_{0|h00} - v''_{0|h} \leq v_{0|h} - v''_{0|h} \leq (t+1) \delta,
    \]
    which implies $|v_{0|h00} - v''_{0|h00}| = |v_{0|h00} - v''_{0|h}| \leq (t+1) \delta$.
    If $v'_{0|h00} <0$, $|v_{0|h00} - v''_{0|h00}| = |v_{0|h00}| \leq |v_{0|h00} - v'_{0|h00}| \leq (t+1) \delta$. Otherwise, $|v_{0|h00} - v''_{0|h00}| = |v_{0|h00} - v'_{0|h00}| \leq \delta$. In summary, we have $|v_{0|h00} - v''_{0|h00}| \leq (t+1) \delta$.

    We next show $|v_{1|h00} - v''_{1|h00}| \leq (t+2) \delta$.  If $v'_{0|h00} > v''_{0|h}$, $|v_{1|h00} - v''_{1|h00}| = v_{1|h00} = v_{0|h} - v_{0|h00} \leq (v''_{0|h} + (t+1)\delta) - (v'_{0|h00} - \delta) \leq (t+2)\delta$. If $v'_{0|h00} <0$, $v''_{1|h00} = v''_{0|h}$ and $|v_{1|h00} - v''_{1|h00}| = |v_{1|h00} - v''_{0|h}| \leq |v_{1|h00} - v_{0|h}| + |v_{0|h} - v''_{0|h}| = |v_{0|h00}| + |v_{0|h} - v''_{0|h}| \leq (t+2)\delta$. Otherwise, $|v_{1|h00} - v''_{1|h00}| = |v_{1|h00} - (v''_{0|h} - v'_{0|h00})| = |(v_{0|h} - v_{0|h00}) - (v''_{0|h} - v'_{0|h00})| \leq |v_{0|h} - v''_{0|h}| + |v_{0|h00} - v'_{0|h00}| \leq (t+2)\delta$.

    Taking all the values $a,b \in \{0,1\}$ into account, we have $|v_{m|hab} - v''_{m|hab}| \leq (t+2) \delta$ for $m \in \{0,1\}$, $h \in \{0,1\}^t$, $a \in \{0,1\}$ and $b \in \{0,1\}$. By induction on $t$, we have the claim.
\end{proof}

For all $h$, $p^B_{0|h} + p^B_{1|h} = 1$. We thus have
\begin{align*}
    &\sum_{h \in \{0,1\}^{2r}} \prod_{i \in [r]} p^B_{h[2i-1]|h[0,2i-2]}  
    = \sum_{h[0] \in \{0,1\}} \sum_{h[1] \in \{0,1\}} \cdots \sum_{h[2r-1] \in \{0,1\}} \prod_{i \in [r]} p^B_{h[2i-1]|h[0,2i-2]} \\
    &= \sum_{h[0] \in \{0,1\}} \sum_{h[2] \in \{0,1\}} \cdots \sum_{h[2r-2] \in \{0,1\}} \sum_{h[1] \in \{0,1\}} p^B_{h[1]|h[0]}  \sum_{h[3] \in \{0,1\}} p^B_{h[3]|h[0,2]} \cdots \sum_{h[2r-1] \in \{0,1\}}  p^B_{h[2i-1]|h[0,2i-2]} \\
    &= \sum_{h[0] \in \{0,1\}} \sum_{h[2] \in \{0,1\}} \cdots \sum_{h[2r-2] \in \{0,1\}} 1 = 2^{r}
\end{align*}

Therefore, the error is at most
\begin{align*}
    &\sum_{h \in \{0,1\}^{2r}} \prod_{i \in [r]} p^B_{h[2i-1]|h[0,2i-2]} (r+1) \delta = 2^r(r+1)\delta,
\end{align*}
which concludes the proof.
\end{proof}

\begin{theorem}\label{thm:two-way_LOCC}
    For a sufficiently small constant $k>0$, $n \in\mathbb{N}$, let $\mathcal{P}$ be a quantum $2$-value $2k \log_2 n$-round LOCC protocol to solve $\EQ_n$. Then the size of Alice's message $|\rho_x|$ is $\omega(n^{1-3k-\epsilon})$ for an arbitrary small constant $\epsilon >0$.
\end{theorem}

\begin{proof}
    Let $\delta$ be $n^{-k-\epsilon_1}$ for a sufficiently small constant $\epsilon_1 >0$. Suppose that $|\rho_x|$ is $O(n^{1-3(k+\epsilon_1)-\epsilon_2})$ for a sufficiently small constant $\epsilon_2 >0$. From \cref{eq:num}, the number of values the referee guesses is $2^{2 k \log_2 n}$. Then, from \cref{lem:replace} and \cref{lem:error}, the length of the deterministic message is $O ( (n^{1-3(k+\epsilon_1)-\epsilon_2}) n^{3(k+\epsilon_1)} ((1-3(k+\epsilon_1)-\epsilon_2)\log n + (k+\epsilon_1) \log n) (2 k \log_2 n + (k+\epsilon_1)\log n) )= o(n)$ and the error is $n^k (k \log n + 1) (n^{-k-\epsilon_1}) = o(1)$. If one message is deterministic in the hybrid SMP model and the model solves $\EQ$ with high probability, the length of the message is $n$. Since we can take $\epsilon = 3\epsilon_1+\epsilon_2$ an arbitrary small value, we have $|\rho_x| = \omega(n^{1-3k-\epsilon})$ for an arbitrary small constant $\epsilon>0$.
\end{proof}

\subsection{General arguments}

In the proofs of \cref{subsec:warm-up} and \cref{subsec:many_rounds}, we use a property that is inherent to $\EQ_n$. The property is that, if one message is deterministic in SMP models for $\EQ_n$, the length of the deterministic message is $n$. We can show more general results by replacing both quantum messages.

\begin{lemma}\label{lem:error_replace_both}
    Let $\mathcal{P}$ be a $2$-value $2r$-round LOCC protocol to solve $F:\{0,1\}^n \times \{0,1\}^n \rightarrow \{0,1\}$. Using a deterministic message from \cref{lem:replace}, the referee can simulate the quantum LOCC SMP protocol with error $O(2^{2r}r^2\delta^2)$.
\end{lemma}

\begin{proof}
For $m \in \{0,1\}$ and $h \in \{0,1\}^*$, let us define $v^A_{m|h} := v_{m|h}, v''^A_{m|h} := v''_{m|h}$ and, in a similar way, define $v^B_{m|h}, v''^B_{m|h}$ from the replacement of Bob's message.
Then, the error by the simulation of the referee from the two deterministic strings is at most
\begin{align*}
    & \left| \sum_{h \in \{0,1\}^{2r}} v^A_{1|h} v^B_{h[2i-1]|h[0,2i-2]} - v''^A_{1|h} v''^B_{h[2i-1]|h[0,2i-2]} \right| \\
    & \leq \sum_{h \in \{0,1\}^{2r}} \left| v^A_{1|h} v^B_{h[2i-1]|h[0,2i-2]} - v''^A_{1|h} v''^B_{h[2i-1]|h[0,2i-2]} \right| \\
    & \leq \sum_{h \in \{0,1\}^{2r}} \left| v^A_{1|h} v^B_{h[2i-1]|h[0,2i-2]} - \left( v^A_{1|h} + O(r\delta) \right) \left( v^B_{h[2i-1]|h[0,2i-2]} + O(r\delta) \right) \right| = O(2^{2r}r^2\delta^2) 
\end{align*}
\end{proof}

\begin{proposition}\label{prop:two-way-LOCC_general}
    For a sufficiently small constant $k>0$, $n \in\mathbb{N}$, let $\mathcal{P}$ be a quantum $2$-value $2k \log_2 n$-round LOCC protocol to solve $F:\{0,1\}^n \times \{0,1\}^n \rightarrow \{0,1\}$ such that $\R^{||}(F) = \Omega(n^c)$ for a constant $c>0$. Then, the size of the quantum messages in $\mathcal{P}$ is $\omega(n^{c-6k-\epsilon})$ for an arbitrary small constant $\epsilon >0$.
\end{proposition}

\begin{proof}
    Let $\delta$ be $n^{-2k-\epsilon_1}$ for a sufficiently small constant $\epsilon_1 >0$. Let $\rho_x$ and $\sigma_y$ be quantum messages from Alice and Bob, respectively. Suppose that $|\rho_x|$ and $|\sigma_y|$ are $O(n^{c-3(2k+\epsilon_1)-\epsilon_2})$ for a sufficiently small constant $\epsilon_2 >0$. From \cref{eq:num}, the number of values the referee guesses is $2^{2 k \log_2 n}$ for each message. Then, from \cref{lem:replace} and \cref{lem:error_replace_both}, the sum of the lengths of the deterministic messages is $2 \times O(n^{c-3(2k+\epsilon_1)-\epsilon_2}) n^{3(2k+\epsilon_1)} ((c-3(2k+\epsilon_1)-\epsilon_2)\log n + (2k+\epsilon_1) \log n) (2 k \log_2 n + (2k+\epsilon_1)\log n)= o(n^c)$ and the error is $O(n^{2k}(k \log n)^2 n^{-2k-\epsilon_1}) = o(1)$. This contradicts $\R^{||}(F) = \Omega(n^c)$. Since we can take $\epsilon = 3\epsilon_1+\epsilon_2$ an arbitrary small value, we have $\max\{|\rho_x|,|\sigma_y|\} = \omega(n^{c-6k-\epsilon})$ for an arbitrary small constant $\epsilon>0$.
\end{proof}

Moreover, we can consider a lower bound for general relation problems.

\begin{lemma}
    Let $\mathcal{P}$ be a $2$-value $2r$-round LOCC protocol to solve $R \subseteq \{0,1\}^n \times \{0,1\}^n \times \{0,1\}^*$. Using a deterministic message from \cref{lem:replace}, the referee can simulate the output of the quantum LOCC SMP protocol with error $O(2^{2r}r^2\delta^2)$.
\end{lemma}

\begin{proof}
For $m \in \{0,1\}$ and $h \in \{0,1\}^*$, let us define $v^A_{m|h} := v_{m|h}, v''^A_{m|h} := v''_{m|h}$ and, in a similar way, define $v^B_{m|h}, v''^B_{m|h}$ from the replacement of Bob's message.
Then, the error by the simulation of all the measurement outcomes of the referee from the two deterministic strings is at most
\begin{align*}
    & \left| \sum_{m \in \{0,1\}, h \in \{0,1\}^{2r}} v^A_{m|h} v^B_{h[2i-1]|h[0,2i-2]} - v''^A_{m|h} v''^B_{h[2i-1]|h[0,2i-2]} \right| \\
    & \leq 2 \times \sum_{h \in \{0,1\}^{2r}} \left| v^A_{1|h} v^B_{h[2i-1]|h[0,2i-2]} - v''^A_{1|h} v''^B_{h[2i-1]|h[0,2i-2]} \right| \\
    & \leq 2 \times \sum_{h \in \{0,1\}^{2r}} \left| v^A_{1|h} v^B_{h[2i-1]|h[0,2i-2]} - \left( v^A_{1|h} + O(r\delta) \right) \left( v^B_{h[2i-1]|h[0,2i-2]} + O(r\delta) \right) \right| = O(2^{2r}r^2\delta^2) 
\end{align*}
If all the measurement outcomes are the same, the output of the protocol is also the same as the original protocol.
\end{proof}

Then, by the same proof as \cref{prop:two-way-LOCC_general}, we have the same argument for general relation problems.

\begin{theorem}\label{thm:two-way-LOCC_general_relation}
    For a sufficiently small constant $k>0$, $n \in\mathbb{N}$, let $\mathcal{P}$ be a $2$-value $2k \log_2 n$-round LOCC protocol to solve $R \subseteq \{0,1\}^n \times \{0,1\}^n \times \{0,1\}^*$ such that $\R^{||}(R) = \Omega(n^c)$ for a constant $c>0$. Then, the size of the quantum messages in $\mathcal{P}$ is $\omega(n^{c-6k-\epsilon})$ for an arbitrary small constant $\epsilon >0$.
\end{theorem}

\section{Separation between quantum one-way and two-way-LOCC SMP protocols}\label{sec:separation}

For hybrid SMP schemes for non-symmetric problems, it is important which party can send a quantum message while the other one can only send classical messages, and let us introduce more fine-grained hybrid SMP models. In this section, we will consider a hybrid SMP model in which Alice sends a quantum message and Bob sends a classical message, and let us denote by $\Q \R^{||}(R)$ the hybrid SMP communication complexity for $R$. We also consider a hybrid SMP model in which Bob sends a quantum message and Alice sends a classical message, and let us denote by $\R \Q^{||}(R)$ the hybrid SMP communication complexity for $R$. 
We also define $\Q^{||,\mathsf{LOCC}_1^{A \rightarrow B}}(R)$ as the complexity of a quantum one-way LOCC protocol for $R$ where $\rm{Ref}_A$ measures Alice's message first and sends a measurement result to $\rm{Ref}_B$, and then $\rm{Ref}_B$ measures Bob's message, $\Q^{||,\mathsf{LOCC}_1^{B \rightarrow A}}(R)$ as the complexity of a quantum one-way LOCC protocol for $R$ where the $\rm{Ref}_B$ measures Bob's message first and sends a measurement result to $\rm{Ref}_A$, and then $\rm{Ref}_A$ measures Alice's message.
Other notations of communication complexity are the same as defined in \cref{sec:prel}.

The Hidden Matching Problem was introduced by Bar-Yossef, Jayram, and Kerenidis \cite{BYJK08} and they showed that there exists a relation problem which exhibits an exponential separation between classical and quantum one-way communication complexity. For $k \in [n]$, let us denote by $x(k)$ the $(k-1)$th bit of $x \in \{0,1\}^n$.

\begin{definition}[The Hidden Matching Problem  ($\mathsf{HM}_n$) \cite{BYJK08}]
    Let $n$ be an even positive integer. In the Hidden Matching Problem ($\mathsf{HM}_n$), Alice is given $x \in \{0,1\}^n$ and Bob is given $M \in \mathcal{M}_n$ where $ \mathcal{M}_n$ is the family of all possible perfect matchings on $n$ nodes. They coordinate to output a tuple $\langle i,j,b \rangle$ such that the edge $(i,j)$ belongs to the matching $M$ and $b = x(i) \oplus x(j)$.
\end{definition}

\begin{theorem}[\cite{BYJK08}]
    $\Q^{1}(\mathsf{HM}_n) = O(\log n)$, $\R^{1}(\mathsf{HM}_n) = \Theta(\sqrt{n})$.
\end{theorem}

For completeness, let us describe a quantum protocol to solve $\mathsf{HM}_n$ efficiently.

\begin{algorithm}[H]
\caption{\, $O(\log n)$ qubit one-way communication protocol for $\mathsf{HM}_n$ \cite{BYJK08}}
\label{}
\begin{algorithmic}[1]
\State Alice sends $O(\log n)$-qubit state $\frac{1}{\sqrt{n}} \sum_{i=1}^n (-1)^{x(i)} \ket{i}$ to Bob. 
\State Bob regards $(i,j) \in M$ as a projector $P_{ij} = \ket{i}\bra{i} + \ket{j}\bra{j}$, and conducts the corresponding projective measurement on the quantum message. If the result is $(i,j) \in M$, Bob finally measures the residue state in a base $\{ \frac{1}{\sqrt{2}} (\ket{i}+\ket{j}), \frac{1}{\sqrt{2}}(\ket{i}-\ket{j}) \}$ and outputs the final measurement result.
\end{algorithmic}
\end{algorithm}

They also introduced a variant of the Hidden Matching Problem and showed a separation between classical SMP models and the hybrid SMP model.

\begin{definition}[Restricted Hidden Matching Problem ($\mathsf{RHM}_n$), Section 5 in \cite{BYJK08}]
    Let $n$ be an even positive integer. In the Restricted Hidden Matching Problem, fix $\mathcal{M}$ to be any set of $m = \Theta(n)$ pairwise edge-disjoint matchings. Alice is given $x \in \{0, 1\}^n$, and Bob is given $M \in \mathcal{M}$. Their goal is to output a tuple $\langle i, j, b \rangle$ such that the edge $(i, j)$ belongs to the matching $M$ and $b = x(i) \oplus x(j)$.
\end{definition}

\begin{lemma}[Section 5 in \cite{BYJK08}]\label{lem:rhm}
    $\R^{1}(\mathsf{RHM}_n) = \R^{||}(\mathsf{RHM}_n) = \Theta(\sqrt{n})$, $\Q\R^{||}(\mathsf{RHM}_n) = O(\log n)$.
\end{lemma}

Note that in the hybrid SMP protocol, Bob tells the referee which matching he received from $O(n)$ matchings by $O(\log n)$ bits.

We show that, in a quantum one-way LOCC protocol, if the referee measures Bob's message first, the problem is easy, and, if the referee measures Alice's message first, the problem is hard.

\begin{proposition}
    $\Q^{||,\mathsf{LOCC}_1^{B \rightarrow A}}(\mathsf{RHM}_n) = O(\log n)$, $\R\Q^{||}(\mathsf{RHM}_n) = \Q^{||,\mathsf{LOCC}_1^{A \rightarrow B}}(\mathsf{RHM}_n) = \Theta(\sqrt{n})$.
\end{proposition}

\begin{proof}
    $\Q^{||,\mathsf{LOCC}_1^{B \rightarrow A}}(\mathsf{RHM}_n) = O(\log n)$ follows from $\Q\R^{||}(\mathsf{RHM}_n) = O(\log n)$ by encoding and decoding the information of the matchings in the computational basis.
    
    When we regard Bob and the referee as one party and allow arbitrary communication between them, the communication model is the same as the classical one-way communication model. In the classical one-way communication model, $\Theta(\sqrt{n})$ bits communication is required from \cref{lem:rhm}. Therefore, in the hybrid SMP model as a weaker model than the classical one-way communication model, $\Theta(\sqrt{n})$ bits communication is also required (i.e, $\R\Q^{||}(\mathsf{RHM}_n) = \Theta (\sqrt{n})$).
    
    Finally, consider $\Q^{||,\mathsf{LOCC}_1^{A \rightarrow B}}(\mathsf{RHM}_n)$. From \cref{thm:one-way_LOCC}, we can replace Alice's quantum message with a classical message by a small overhead, and we have $\Q^{||,\mathsf{LOCC}_1^{A \rightarrow B}}(\mathsf{RHM}_n) =  \Theta (\sqrt{n})$.
\end{proof}

We next consider a further variant of the Restricted Hidden Matching Problem to make both parties send quantum messages.

\begin{definition}[Double Restricted Hidden Matching problem ($\mathsf{DRHM}_n$]
    Let $n$ be an even positive integer. Take two independent instances from the Restricted Hidden Matching Problem. Let $x_1 \in \{0, 1\}^n, M_1 \in \mathcal{M}_1$ be an input of the first instance and $x_2 \in \{0, 1\}^n, M_2 \in \mathcal{M}_2$ be an input of the second instance. Alice is given $x_1$ and $M_2$ and Bob is given $x_2$ and $M_1$. Their goal is to output two tuples $\langle i_1, j_1, b_1 \rangle$ and $\langle i_2, j_2, b_2 \rangle$ to solve the two Restricted Hidden Matching Problem.
\end{definition}

\begin{proposition}\label{prop:separation}
    $\Q^{||,\mathrm{LOCC}}(\mathsf{DRHM}_n) = O(\log n)$, $\R^{||} (\mathsf{DRHM}_n) = \R\Q^{||} (\mathsf{DRHM}_n) = \Q\R^{||} (\mathsf{DRHM}_n) = \Q^{||,\mathsf{LOCC}_1^{B \rightarrow A}}(\mathsf{DRHM}_n) = \Q^{||,\mathsf{LOCC}_1^{A \rightarrow B}}(\mathsf{DRHM}_n) =  \Theta(\sqrt{n})$.
\end{proposition}

\begin{proof}
    Let us first show $\Q^{||,\mathrm{LOCC}}(\mathsf{DRHM}_n) = O(\log n)$. Alice sends $\rm{Ref}_A$ $\frac{1}{\sqrt{n}} \sum_{i=1}^n (-1)^{x_{1}(i)} \ket{i}$ and $\ket{M_2}$ where $M_2$ is encoded in the computational basis. Bob sends $\rm{Ref}_B$ $\frac{1}{\sqrt{n}} \sum_{i=1}^n (-1)^{x_{2}(i)} \ket{i}$ and $\ket{M_1}$ where $M_1$ is encoded in the computational basis. Then, $\rm{Ref}_A$ measures $\ket{M_2}$ and sends $M_2$ to $\rm{Ref}_B$. $\rm{Ref}_B$ measures $\ket{M_1}$ and sends $M_1$ to $\rm{Ref}_A$. Based on $M_1$, $\rm{Ref}_A$ measures $\frac{1}{\sqrt{n}} \sum_{i=1}^n (-1)^{x_{1}(i)} \ket{i}$. Based on $M_2$, $\rm{Ref}_B$ measures $\frac{1}{\sqrt{n}} \sum_{i=1}^n (-1)^{x_{2}(i)} \ket{i}$. Note that each matching can be represented with $O(\log n)$ bits because the size of the set of matchings $|\mathcal{M}_1|$ and $|\mathcal{M}_2|$ is $\Theta(n)$.
    
    Let us next show $\R\Q^{||} (\mathsf{DRHM}_n) = \Theta(\sqrt{n})$. When we regard Bob and the referee as one party and allow arbitrary communication between them, the communication model is the same as the classical one-way communication model. The first and second instances are independently chosen, and to solve the first instance, $\Theta(\sqrt{n})$ bits communication is required in the model from \cref{lem:rhm}. Therefore, in the SMP model as a weaker model than the classical one-way communication model, $\Theta(\sqrt{n})$ bits communication from Alice to the referee is required. By a similar discussion, it is shown that $\Q\R^{||} (\mathsf{DRHM}_n) = \Theta(\sqrt{n})$.

    Finally, from \cref{thm:one-way_LOCC}, we have $\Q^{||,\mathsf{LOCC}_1^{B \rightarrow A}}(\mathsf{DRHM}_n) = \Q^{||,\mathsf{LOCC}_1^{A \rightarrow B}}(\mathsf{DRHM}_n) = \Theta(\sqrt{n})$.
\end{proof}

\begin{remark}\label{remark}
    In the quantum two-way-LOCC protocol, we can allow $\rm{Ref}_A$ and $\rm{Ref}_B$ to do several 2-value measurements in sequence (by making the other one do nothing in the turns). For $\mathsf{DRHM}_n$, we can solve it by $2 \log_2 n + O(\log n)$ rounds 2-value two-way-LOCC measurements. Let us describe the protocol. First, $\rm{Ref}_A$ measures $\ket{M_2}$ in the computational basis, which takes $O(\log n)$ rounds. Similarly, $\rm{Ref}_B$ measures $\ket{M_1}$ that takes $O(\log n)$ rounds. Then, based on $M_1$,  $\rm{Ref}_A$ puts half of ${(i,j) \in M_1}$ into one group and the other half into another, and measures $\frac{1}{\sqrt{n}} \sum_{i=1}^n (-1)^{x_{1}(i)} \ket{i}$ by a projector constructed from the grouping. Based on the measurement outcome, $\rm{Ref}_A$ repeats another 2-value measurement with a similar half/half grouping. $\rm{Ref}_B$ conducts a similar sequence of 2-value measurements on $\frac{1}{\sqrt{n}} \sum_{i=1}^n (-1)^{x_{2}(i)} \ket{i}$ based on $M_2$.

    However, the number of rounds is larger than the bound obtained from \cref{thm:two-way-LOCC_general_relation}, and thus there is no contradiction. Moreover, this implies that our analysis is tight up to some constant factor. Therefore, to prove a stronger lower bound toward \cref{conj}, we should develop another proof strategy which might be inherent to $\EQ_n$.
\end{remark}

\section*{Acknowledgment}
AH thanks Richard Cleve, Uma Girish, Alex B. Grilo, Kohdai Kuroiwa, Alex May, Masayuki Miyamoto, Ryuhei Mori, Ashwin Nayak, Daiki Suruga, Yuki Takeuchi, Eyuri Wakakuwa, Yibin Wang for helpful discussions. AH is also grateful to Richard Cleve and Ken-ichi Kawarabayashi for their generous support. Part of the work was done while AH was visiting University of Waterloo, and AH is grateful for their hospitality.

AH was supported by JSPS KAKENHI grants Nos.~JP22J22563, 24H00071, and  JST ASPIRE Grant No.~JPMJAP2302. SK was supported by the Natural Sciences and Engineering Research Council of Canada (NSERC) Discovery Grants Program, and Fujitsu Labs America. FLG was supported by JSPS KAKENHI grants Nos.~JP20H05966, 20H00579, 24H00071, MEXT Q-LEAP grant No.~JPMXS0120319794 and JST CREST grant No.~JPMJCR24I4. HN was supported by JSPS KAKENHI grants Nos.~JP20H05966, 22H00522, 24H00071, 24K22293, MEXT Q-LEAP grant No.~PMXS0120319794 and JST CREST grant No.~JPMJCR24I4. QW was supported by the Engineering and Physical Sciences Research Council under Grant No.~EP/X026167/1.

\bibliographystyle{alpha}
\bibliography{ref}

\end{document}